\documentclass[sigconf,table]{acmart}
\AtBeginDocument{%
  }

\usepackage{style}

  \newcommand{\emp}[1]{\emph{\textbf{#1}}}

\definecolor{forestgreen}{rgb}{0.13, 0.55, 0.13}

\newcommand{\rctreetracing}{{\textsc{RCTreeTracing}}\xspace}

\newcommand{\rctt}{\ensuremath{\mathsf{RCTT}}\xspace}
\newcommand{\paruf}{\ensuremath{\mathsf{ParUF}}\xspace}
\newcommand{\sequf}{\ensuremath{\mathsf{SeqUF}}\xspace}
\newcommand{\spine}{\ensuremath{\mathsf{spine}}\xspace}

\newcommand{\bottomupdend}{{\textsc{SLD-Activation}}\xspace}
\newcommand{\sldmerge}{{\textsc{SLD-Merge}}\xspace}
\newcommand{\sldtreecon}{{\textsc{SLD-TreeContraction}}\xspace}

\def \compress {\ensuremath{\mathtt{{compress}}}}

\def \rake {\ensuremath{\mathtt{{rake}}}}
\def \rcnode {\mathtt{rcnode}}

\newcommand{\CAS}{{\textsc{CAS}}\xspace}
\newcommand{\ati}{{\textsc{atomic\_inc}}}

\newcommand{\R}{\mathbb{R}}

\newcommand{\ifconference}{{{\ifx\fullversion\undefined}}}

\makeatletter
\def\dfnt@space@setup{%
\dfnt@preskip=\parskip
  \dfnt@postskip=0pt}
\makeatother

\newtheoremstyle{exampstyle}
{.05in} 
{.05in} 
{} 
{.5em} 
{\sc \bfseries} 
{.} 
{.5em} 
{} 
\theoremstyle{exampstyle} 
\theoremstyle{exampstyle} 

\makeatletter
\renewenvironment{proof}[1][\proofname]{\par
\vspace{-\topsep}
\pushQED{\qed}%
\normalfont
\topsep0pt \partopsep0pt 
\trivlist
\item[\hskip\labelsep
      \itshape
  #1\@addpunct{.}]\ignorespaces
}{%
\popQED\endtrivlist\@endpefalse
\addvspace{3pt plus 3pt} 
}

 \crefname{section}{Sec.}{Sec.}
 \crefname{theorem}{Thm.}{Thm.}
 \crefname{lemma}{Lem.}{Lem.}
 \crefname{corollary}{Col.}{Col.}
 \crefname{table}{Tab.}{Tabs.}
 \crefname{algorithm}{Alg.}{Algs.}
 \Crefname{table}{Tab.}{Tabs.}
\crefname{claim}{Claim}{Claim}

\begin{document}
\fancyhead{}
\newtheorem{claim}[theorem]{Claim}
\title{Optimal Parallel Algorithms for \texorpdfstring{\\}{}
Dendrogram Computation and Single-Linkage Clustering}

\author{Laxman Dhulipala}
\affiliation{
    \institution{University of Maryland}
   \city{College Park}
   \state{MD}
   \country{USA}
}
\email{laxman@umd.edu}

\author{Xiaojun Dong}
\affiliation{
    \institution{University of California}
   \city{Riverside}
   \state{CA}
   \country{USA}
}
\email{xdong038@ucr.edu}

\author{Kishen N Gowda}
\affiliation{
    \institution{University of Maryland}
   \city{College Park}
   \state{MD}
   \country{USA}
}
\email{kishen19@cs.umd.edu}

\author{Yan Gu}
\affiliation{
    \institution{University of California}
   \city{Riverside}
   \state{CA}
   \country{USA}
}
\email{ygu@cs.ucr.edu}


\begin{abstract}
Computing a Single-Linkage Dendrogram (SLD) is a key step in the 
classic single-linkage hierarchical clustering algorithm.
Given an input edge-weighted tree $T$, the SLD of $T$ is a binary
dendrogram that summarizes the $n-1$ clusterings 
obtained by contracting the edges of $T$ in order of weight.
Existing algorithms for computing the SLD all require $\Omega(n\log n)$
work where $n = |T|$. Furthermore, to the best of our knowledge no
prior work provides a parallel algorithm obtaining non-trivial speedup
for this problem.

In this paper, we design faster parallel algorithms for computing SLDs
both in theory and in practice based on new structural results about SLDs.
In particular, we obtain a deterministic output-sensitive parallel algorithm based 
on parallel tree contraction that requires $O(n \log h)$ work and 
$O(\log^2 n \log^2 h)$ depth, where $h$ is the height of the output SLD.
We also give a deterministic bottom-up algorithm for the problem inspired by 
the nearest-neighbor chain algorithm for hierarchical agglomerative 
clustering, and show that it achieves $O(n\log h)$ work and $O(h \log n)$ depth.
Our results are based on a novel divide-and-conquer 
framework for building SLDs, inspired by 
divide-and-conquer algorithms for Cartesian trees.
Our new algorithms can quickly compute the SLD on billion-scale 
trees, and obtain up to 150x speedup over the highly-efficient Union-Find 
algorithm typically used  to compute SLDs in practice.

\end{abstract}


\hide{

\ccsdesc[500]{Theory of computation~Shared memory algorithms}
\ccsdesc[500]{Theory of computation~Parallel algorithms}
\ccsdesc[500]{Theory of computation~Graph algorithms analysis}

\keywords{Single-Linkage Clustering, Hierarchical Graph Clustering, HAC, Dendrograms, Parallel Algorithms}
}

\renewcommand\footnotetextcopyrightpermission[1]{} 


\maketitle

\section{Introduction}
Single-linkage clustering is a fundamental technique in unsupervised
learning and data mining that groups objects based on a similarity (or
dissimilarity) function~\cite{schutze2008introduction}.
The single-linkage clustering of an edge-weighted tree $T$ is defined
as the tree of clusters (a dendrogram) obtained by sequentially
merging the edges of $T$ in decreasing (increasing) order of
similarity (dissimilarity) as follows:
\begin{enumerate}[label=(\arabic*),topsep=0pt,itemsep=0pt,parsep=0pt,leftmargin=20pt]
\item place each vertex of $T$ in its own cluster
\item sort the edges of $T$ by weight
\item for each edge $(u,v)$ in sorted order, merge the clusters of
$u$ and $v$ to form a new cluster.
\end{enumerate}
The output of this clustering process is called
the \defn{single-linkage dendrogram (SLD)}. The SLD is a binary tree
of clusters, where the vertices of $T$ are the leaf clusters, and the
internal nodes correspond to merging two clusters by contracting an edge.
The dendrogram enables users to easily process, visualize, and analyze
the $n-1$ different clusterings induced by the single-linkage hierarchy.

Since hierarchical structure frequently occurs in real world data, single-linkage
dendrograms have been widely used to analyze real-world data in
fields ranging from computational biology~\cite{yengo2022saturated, gasperini2019genome, letunic2007interactive}, 
image analysis~\cite{ouzounis2012alpha, gotz2018parallel, havel2019efficient}, and astronomy~\cite{baron2019machine, feigelson1998statistical}, among 
many others~\cite{henry2005cluster, yim2015hierarchical, irbook}.
%
Due to its real-world importance, and its importance as a sub-step in other fundamental
clustering algorithms such as HDBSCAN*~\cite{campello2015hierarchical, WangEtAl21}, computing
the SLD of an input weighted tree has been widely studied by parallel
algorithms researchers in recent years, with novel algorithms and implementations being
proposed for the shared-memory setting~\cite{havel2019efficient, WangEtAl21}, GPUs~\cite{nolet2023cuslink, sao2024pandora}, and
distributed memory settings~\cite{hendrix2013scalable, gotz2018parallel}.

In this paper, we are interested in {\em both} theoretically-efficient
and practically-efficient
parallel algorithms for computing the single-linkage dendrogram.
Sequentially, a simple and practical SLD algorithm is to faithfully
simulate the specification by essentially running
Kruskal's minimum spanning tree algorithm using Union-Find.
In this algorithm, which we call \defn{\sequf{}},
the $m$ tree edges are first sorted by weight. 
Then, the algorithm runs in $m$ sequential iterations, 
where the $i$-th iteration takes the $i$-th edge and merges the 
clusters corresponding to the endpoints of the edge.
Maintaining information about the current
cluster for a node in the tree is done using Union-Find.
Overall, the work is $O(n\log n)$ due to sorting the edges, and the algorithm has $\Omega(n)$ depth (longest dependence chain) since the edges are merged sequentially.

Wang et al.~\cite{WangEtAl21} recently gave the first work-efficient
parallel algorithm for the problem---their algorithm
computes the SLD in $O(n \log n)$ expected work 
and $O(\log^2 n \log\log n)$ depth with high probability (\whp{}).
Although their algorithm is work-efficient with
respect to \sequf{}, it is challenging to implement as it
relies on applying divide-and-conquer over the weights, which is 
implemented using the Euler Tour Technique~\cite{JaJa92}. 
Due to its complicated nature, this algorithm does not consistently outperform the simple \sequf{}.
Thus, the authors only released the code for \sequf{} and suggested to always use \sequf{} rather than the theoretically-efficient algorithm.
The algorithm is also randomized due to the use of semisort~\cite{WangEtAl21, GSSB15}, and there is no obvious way to derandomize it to obtain a deterministic parallel algorithm for the problem.
%



As a fundamental problem on trees with significant real-world applicability, 
an important question then is whether there is a {\em relatively simple} parallel 
algorithm for this problem that has good theoretical guarantees, is more 
readily implementable, and can obtain non-trivial speedups over \sequf{}.
In this paper, we give a strong positive answer to this question by giving two 
new algorithms for SLD computation that achieve consistent and large 
speedups over \sequf{}, both of which have good theoretical guarantees.
Both of our algorithms are obtained through a better understanding of
the structure of SLDs, and in particular, by showing how to build SLDs 
in a divide-and-conquer fashion using a {\em merge} primitive that can
merge the SLDs of two trees under certain conditions (\cref{sec:merging_dendrograms}).

We leverage these structural results to design two novel deterministic parallel single-linkage dendrogram algorithms.
The first algorithm is based on parallel tree contraction and stores
spines (node-to-root paths in the dendrogram) in meldable heaps; it uses
heap meld and filter operations to implement dendrogram merging (\cref{sec:tree_contraction})
in the rake and compress steps.
We describe a modification of this algorithm that eliminates the need
for meldable heaps, which makes our algorithm simple to describe and implement (\cref{sec:rctt}).
The second algorithm, which we call \paruf{} is a bottom-up algorithm inspired by the nearest-neighbor chain
algorithm for hierarchical agglomerative clustering that merges all local-minima (edges that are merged before all other neighboring edges by the 
sequential algorithm) in parallel (\cref{sec:paruf}).
We design an asynchronous version of \paruf{}, whose 
practical success shows that there is ample parallelism in many
instances that the \sequf{} algorithm does not exploit.

Our theoretical analysis of both algorithms reveals that the $\Theta(n\log n)$ solution obtained by
existing SLD algorithms is in fact sub-optimal in many cases; both of our algorithms deterministically run in
$O(n\log h)$ work where $h$ is the height of the output dendrogram, where
$\lfloor \log n \rfloor \leq h \leq n-1$.
Thus, our algorithms can require asymptotically less work when the output SLD is not highly skewed. 
For instance, when $h = O(\log n)$ as in the case of a balanced dendrogram, our 
algorithms only use $O(n\log\log n)$ work.
We complement our upper bounds with a simple comparison-based lower bound 
showing that our work bounds are optimal.
Our algorithms are readily implementable and enable us to consistently
compute the dendrogram of a billion-node tree in roughly 10 seconds on a
96-core machine, achieving up to 149x speedup over the highly-optimized
\sequf{} implementation.

The major contributions of this paper are:

\setlength{\itemsep}{0pt}
\begin{itemize}[topsep=1.5pt, partopsep=0pt,leftmargin=*]

  \item A novel merge-based framework for computing the single-linkage
  dendrogram (\cref{sec:merging_dendrograms}), and two instantiations of this framework, the first using parallel tree contraction with meldable heaps (\cref{sec:tree_contraction}), and the second algorithm (called \rctt{}) using the RC-tree tracing technique (\cref{sec:rctt}), which is readily implementable.

  \item \paruf{}, a bottom-up algorithm that is a natural parallelization of the \sequf{} algorithm, and is readily implementable using a fast asynchronous approach (\cref{sec:paruf}).

  \item Analyses of our algorithms showing that the heap-based algorithm (\cref{sec:tree_contraction}) and \paruf{} (\cref{sec:paruf}) deterministically achieve $O(n\log h)$ work, which is optimal for comparison-based algorithms. Our heap-based algorithm has poly-logarithmic depth.


  \item An experimental study of our \paruf{} and \rctt{} algorithms
  showing that they achieve between 2.1--150x
  speedup over \sequf{} on a collection of billion-scale input trees (\cref{sec:experiments}). Our implementation can be found at \url{https://github.com/kishen19/ParSLD}.

\end{itemize}

\hide{
Our algorithms are obtained by exploiting the structure of both the
input tree (unrooted) and the output dendrogram (a rooted tree, where
internal nodes correspond to the edges of the input tree).
We show how the input and output trees are related by designing a
\emp{merge} procedure that takes a $(u,v, w_{uv})$ edge, recursively
computes the dendrograms incident to either endpoint, $D_u, D_v$, and
merges the two dendrograms to obtain the output dendrogram $D$.
The key observation is that the only edges in the SLD that can have
their parent pointer change are those on the \emp{spine} of the $(u,v,
w_{uv})$ edge (see Figure~\ref{todo}).
The idea is similar to the idea of left and right spines in an
existing parallel algorithm for the Cartesian tree problem, which
inspired this work~\cite{BlellochS11}.
However, leveraging this structural fact algorithmically still poses
several non-trivial challenges: how should we structure edges to
merge, and how should the spines be maintained and pruned?

We instantiate our ideas using two quite different algorithms, both of which we prove perform optimal work.
Our first algorithm is an ``activation-based'' algorithm that is
similar in spirit to the nearest-neighbor chain algorithm for
hierarchical agglomerative clustering~\cite{benzecri1982construction}.
We leverage a local checkable condition for determining whether an
edge is {\em ready} to be merged; in each step, the algorithm
identifies and merges all ready edges
We show that this approach, which we call \bottomupdend{} runs in
$O(n\log h)$ work and $O(h \log h)$ depth, where $h \in [\log n, n]$
is the height of the output dendrogram.
We also design an heuristic for overcoming high depth dendrograms,
which we find to be important in some cases in practice.
We provide an optimized implementation of our algorithm using
mergeable heaps, and show that it obtains between X--Yx speedup over
the Kruskal-based algorithm.


\todo[inline]{Kishen: Modify the below result after verifying the binomial heaps based 
idea gives optimal work bounds}

On the theoretical side, our first result improves over Kruskal's
algorithm in terms of work, but can run essentially sequentially in
principle in some adversarial cases.
To understand if we can simultaneously achieve $o(n\log n)$ work and
poly-logarithmic depth, we design a different parallel algorithm that
exploits the structure of the input tree instead of the output tree.
In particular, our algorithm uses a heavy-light decomposition of the
input tree, and a novel merge procedure to efficiently generate
the SLD in parallel.
Our algorithm works by computing partial SLDs at the subtrees rooted
at each node, and maintaining the ``spine'' of each partial SLD, which
preciscely captures all clusters whose parents in the SLD have not yet
been determined.
%
%
Our approach is inspired by viewing the dendrogram as a {\em cartesian
tree} defined over trees instead of sequences.
Inspired by an earlier algorithm for computing the cartesian tree in
parallel~\cite{SB14cartesian}, we present a primitive for {\em
merging} two dendrograms that share a single common vertex and show
how to instantiate this idea using heavy-light decomposition.
Overall, we obtain an algorithm that computes the overall dendrogram
in $O(n\log^{2} h)$ work and $O(\log^2n \log h)$ depth.

To evaluate our ideas, we implemented several sequential and parallel
algorithms for dendrogram construction, and study their scalability
and performance on trees, graphs, and pointsets, with
 billions of nodes on a relatively modest shared-memory machine
equipped with 1TB of RAM.
We implement our algorithms in C++ using a widely-used parallel
toolkit, and will make our implementations publicly-available on
Github.
Talk about implementation / experimental results.

Due to the limit on space, we submit our appendix in the supplementary
material. Our contributions are:

\setlength{\itemsep}{0pt}
\begin{itemize}[topsep=1.5pt, partopsep=0pt,leftmargin=*]
  \item \bottomupdend{}, a simple dendrogram construction algorithm
  that can be practically implemented and improves on the work of the
  state-of-the-art dendrogram solution.

  \item \sldmerge{}, a more sophisticated algorithm based on merging
  dendrogram spines, which runs in $O(n\log^{2} h)$ work and
  $O(\log^2n \log h)$ depth.

  \item A lower-bound in the comparison model showing that our output-sensitive algorithm achieves optimal work.

  \item An experimental study of parallel dendrogram construction. Our
  new implementation offers significant improvements in X...

  \item An end-to-end evaluation of our improved algorithm in
  real-world instances of single-linkage clustering of graphs and
  pointsets.
\end{itemize}
}

%
%
%
%

\section{Preliminaries}\label{sec:prelims}

\myparagraph{Notation}
We denote a graph by $G(V, E)$, where $V$ is the set of vertices and
$E$ is the set of edges in the graph. For weighted graphs, the edges
store real-valued weights. We denote a weight or similarity of an edge
$e = (x, y)$ either by writing $w(x, y)$ or $w(e)$ where $w: E \to \R$ is a weight
function, or by placing weight $w \in \R$ in a tuple $(\{u, v\}, w)$ or $(e, w)$.
%
Given two graphs $G_1(V_1,E_1)$ and $G_2(V_2,E_2)$, $G_1 \cup G_2$ denotes the graph $G(V_1\cup V_2,E_1\cup E_2)$. We also use the notation $G\cup \{e\}$, where $e=(u,v)$ is an edge, to denote the graph $G(V\cup \{u,v\},E\cup\{e\})$.
The number of vertices in a graph is $n = |V|$, and the number of
edges is $m = |E|$. 
$\on{Cut}(X, Y)$ denotes the set of edges between two sets
of vertices $X$ and $Y$.

\myparagraph{Model}
We analyze the theoretical efficiency of our parallel algorithms in
the binary-forking model~\cite{blelloch2019optimal}, a concrete
\emph{work-depth model} used to analyze many recent modern parallel
algorithms.
The model is defined in terms of two complexity measures \defn{work}
and \defn{depth}~\cite{blelloch2019optimal, JaJa92,CLRS}.
The \defn{work} is the total number of operations executed by the
algorithm. The \defn{depth} is the longest chain of sequential
dependencies. We assume that concurrent reads and writes are supported
in $O(1)$ work/depth. A \defn{work-efficient} parallel algorithm is
one with work that asymptotically matches the best-known sequential
time complexity for the problem.  
%




\subsection{Parallel Tree Contraction}
Given a tree $G$, the parallel tree contraction framework by Miller and Reif~\cite{miller1985parallel} contracts $G$ to a single node (or cluster) by repeatedly applying alternate rounds of \rake{} and \compress{}. 

\begin{itemize}[topsep=0pt,itemsep=0pt,parsep=0pt,leftmargin=15pt]
\item[$\blacktriangleright$]{\bf\boldmath$\rake{}(u,v)$:} Given a vertex $v$ of degree $1$ and its (only) neighbor $u$, contract $v$ and merge it into $u$. 

\item[$\blacktriangleright$]{\bf\boldmath$\compress{}(u,v,w)$:} Given a vertex $v$ of degree $2$ and its neighbors $u$ and $w$, contract $v$ and merge it into $u$ (arbitrarily), and make $u$ and $w$ neighbors with $w(u,w) = w(v,w)$. 
\end{itemize}

Parallel tree contraction has been studied since the 1980s~\cite{reif94treecontraction, miller1985parallel,abrahamson1989simple,shun2014sequential,hajiaghayi2022adaptive,acar2017brief,anderson2023parallel}
and can be solved deterministically 
in $O(n)$ work and $O(\log^2 n)$ span in the binary-forking
model by simulating the PRAM algorithm of Gazit et al.~\cite{gazit1988optimal}.
The output of parallel tree contraction assigns each vertex to one of $O(\log n)$ rounds, specifying whether it is raked/compressed, and the edge(s) it is raked/compressed along.
%
%
Another way of viewing the output is as a well-defined hierarchical decomposition (clustering) of trees. Starting with each vertex as a singleton cluster, each rake or compress merges two clusters along an edge; thus the subgraph induced on the vertices in a cluster will always be a connected subtree.
Structurally, it can be represented as a rooted tree known as the rake-compress tree (or RC-tree). 
The RC-tree is essentially a hierarchical clustering defined over the input. 
In RC-trees, we have a node (which we call $\rcnode$) corresponding to each vertex in the input tree. 
If a vertex $v$ is contracted, say via the edge $e=(u,v)$, then the parent of $\rcnode(v)$ will be $\rcnode(u)$, and we also associate the edge $e$ to $\rcnode(v)$.

\begin{figure}[t]
    \centering
    \includegraphics[width=0.375\textwidth]{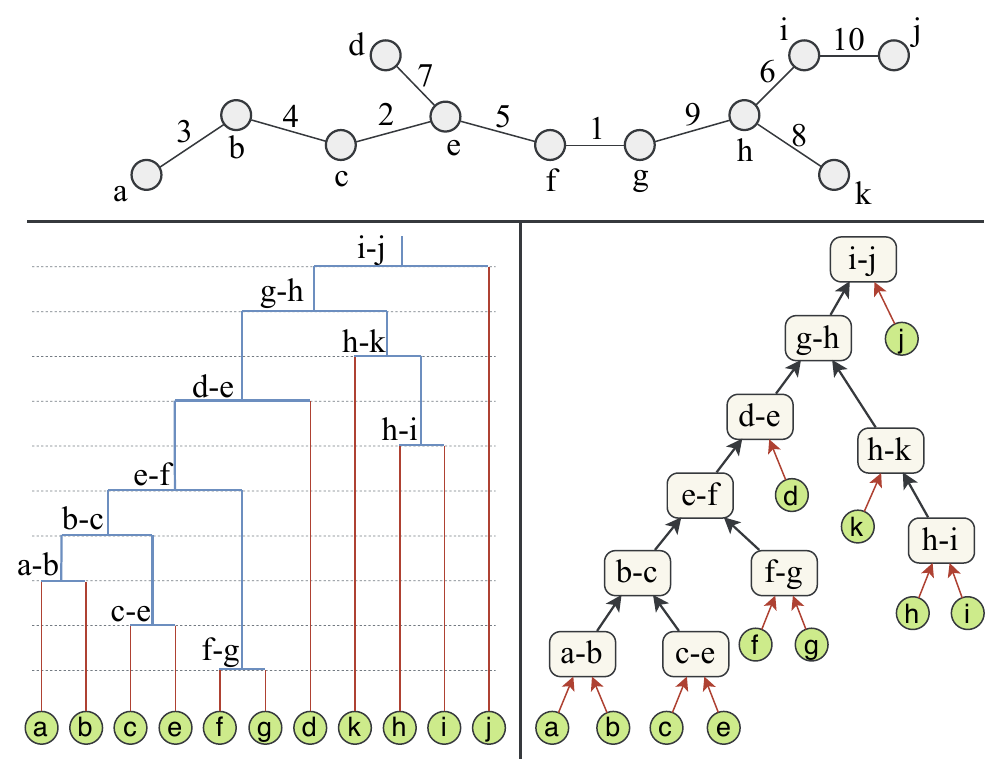}
    \caption{\small Example of single-linkage clustering on the input tree shown in the top panel. The bottom left panel shows a typical visualization of the dendrogram based on the ``height'' of each edge, and the bottom right panel shows the structure of the output SLD.}
    \label{fig:sld_example}
    \Description{Example of single-linkage clustering on the input tree shown in the top panel. The bottom left panel shows a typical visualization of the dendrogram based on the ``height'' of each edge, and the bottom right panel shows the structure of the output SLD.}
    \vspace{-1.5em}
\end{figure}

\subsection{Meldable Heaps}\label{sec:heaps}

We make use of {\em meldable} heaps in this paper. The heaps contain edges keyed
by a comparable edge weight where the comparison function orders edges in sorted
order of rank. Concretely, we make use of binomial heaps~\cite{CLRS}, which support
the {\em meld} operation in logarithmic time. The primitives we use are:
\begin{itemize}[topsep=0pt,itemsep=0pt,parsep=0pt,leftmargin=15pt]
\item $H' \gets \textsc{Insert}(H, e)$: insert $e$ into $H$; return the new heap $H'$.
\item $(H', e) \gets \textsc{DeleteMin}(H)$: remove the minimum element.
\item $H' \gets \textsc{Meld}(H_1, H_2)$: meld the two heaps $H_1$ and $H_2$.
\item $(S, H') \gets \textsc{Filter}(H, e)$: given a heap $H$ and edge $e$, return a pair $(S, H')$ where $S$ contains all edges smaller than $e$ in $H$, and $H'$ contains all elements greater than $e$ in $H$.
\end{itemize}

The \textsc{filter\_and\_insert} primitive used in \cref{alg:rake_optimized} and \cref{alg:compress_optimized} works by first performing \textsc{Insert}, followed by \textsc{Filter}.

\myparagraph{Basic Operations}
Sequential binomial heaps support performing \textsc{Insert}, \textsc{DeleteMin}, and \textsc{Meld} operations in $O(\log n)$ worst case time where $n$ is the number of total elements in the heap(s)~\cite{CLRS}.

\myparagraph{Filter}
We implement \textsc{Filter} on a heap with $s$ elements
in $O(k \log s)$ work and $O(\log^2 s)$ depth where $k$ is the number of elements extracted by the filter operation as follows. 
We independently filter the $O(\log s)$ roots of the binomial trees stored in the binomial heap in parallel.
To filter a binomial tree, we check whether the root is filtered, marking it if so, and recursively proceed in parallel on all children of the root.
This traversal costs $O(k)$ work and runs in $O(\log^2 s)$ depth; the number of nodes filtered in each tree can be computed in the same bounds by treating the set bits as an augmented value.
Emitting the $k$ removed elements into a single array can also be done in the same bounds by using prefix sums.

To rebuild the binomial trees and restore the invariant of a single binomial tree per-rank, we can first emit the children of all nodes removed in the previous step into a single array, where each subtree is stored along with its associated rank.
Rebuilding the heap can then be done by simply performing the same procedure used to build a binomial heap on the trees in this collection.
In a little more detail, the number of subtrees when we remove $k$ nodes is at most $O(k \log s)$.
We can group these subtrees by rank by sorting using a parallel counting sort, which can sort $N$ elements in the range $[0, M]$ in $O(\log n + M)$ depth~\cite{blelloch18notes}.
After sorting the trees into the $O(\log s)$ ranks, rebuilding the trees can be done within each rank in $O(\log s)$ depth using parallel reduce.
We note that the overall structure we maintain is exactly the same as an ordinary binomial heap; we simply augment the heaps with a parallel filter operation that relies on a parallel rebuilding procedure.

\subsection{Single-Linkage Clustering}
Consider an input weighted undirected graph $G(V,E)$. In \emph{single
linkage} clustering, the similarity $\mathcal{W}(X, Y)$ between two clusters $X$
and $Y$ is the minimum similarity between two vertices in $X$ and $Y$,
i.e., 
\begin{equation*}
    \mathcal{W}(X, Y) = \min_{(x,y) \in \on{Cut}(X, Y)} w(x,y).
\end{equation*}

In this paper, we assume the input graph is an edge-weighted tree as
it is well known that single-linkage clustering on weighted (connected\footnote{If the graph is 
not connected, we can solve single-linkage clustering on each connected component independently.}) 
graphs can be reduced to single-linkage clustering on weighted trees; the edges considered for 
merges are exactly that of the minimum spanning tree of the graph~\cite{Gower1969MST}.

Given an input edge-weighted tree $G(V,E)$, let the rank $r_e \in [n]$ of an edge $e$ be 
the position of this edge in the edge sequence sorted by weight (ties broken consistently). We note that our algorithms do not require us to compute the ranks; however, this simplifies the presentation of our algorithms.

The single-linkage dendrogram (SLD) of a tree $G(V,E)$ is a rooted-tree $D(G)$ where each leaf corresponds 
to a vertex in $V$ and each internal node in $D(G)$ corresponds to an edge in $E$. We 
denote the internal node corresponding to edge $e$ as $node(e)$ or simply node $e$. See \Cref{fig:sld_example} 
for an example.

The SLD problem has been widely-studied in recent years~\cite{campello2015hierarchical, WangEtAl21, havel2019efficient, nolet2023cuslink, sao2024pandora}. However, only Wang et al.'s algorithm~\cite{WangEtAl21} is work-efficient with respect to \sequf{}. We discuss related work on SLDs in more detail in \Cref{sec:related_work}.

We assume the output SLD will be stored as a linked list, where each node $e$ points to its parent node $p(e)$. For convenience, we drop the leaf nodes and only consider the tree on the $n-1$ internal edge nodes. For an edge $e\in E$, the $\spine{}_{D(G)}(e)$ denotes the partial linked list starting from node $e$ until the root in $D(G)$. We use $\spine{}(e)$ when the tree and SLD are clear from context. We use the terms SLD and dendrogram interchangeably.

In this paper, we deal with three types of trees: the input tree, the SLD and RC-trees. To keep things clear, we use \emph{vertices} and \emph{edges} when referring to the input, \emph{nodes} and \emph{links} when referring to the SLD, and $\rcnode$ when referring to RC-trees.

\section{Merge-based Algorithms}\label{sec:merge}
In this section, we discuss merge-based algorithms for computing SLDs. 
We first describe a merge subroutine called \sldmerge that allows us to merge the dendrograms of two subtrees, and thus enables various divide-and-conquer algorithms. 
Leveraging the merge subroutine, we give an optimal algorithm for computing SLDs with the help of the parallel tree contraction framework.

\begin{figure}[!t]
    \centering
    \vspace{-1em}
    \includegraphics[width=0.4\textwidth]{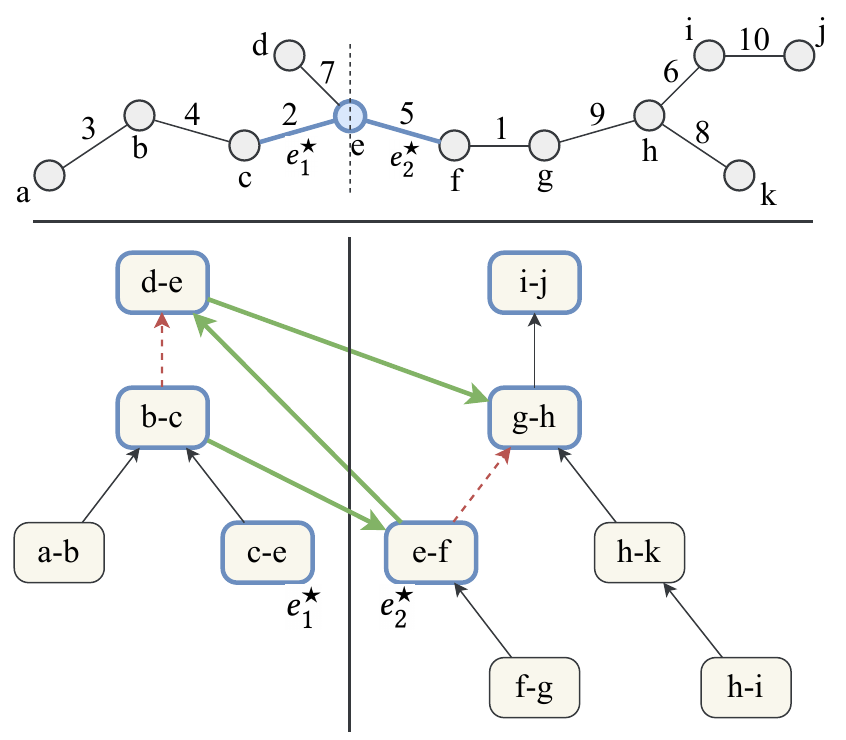}
    \caption{\small An example illustrating \sldmerge{}. The tree is split at node $e$ into two trees (the left
    and right sides of the dashed line) which share no edges, and only share the vertex $e$.
    The \sldmerge{} routine merges the two spines formed by the lowest-rank edge
    incident to $e$ in both trees.
    }
    \label{fig:sld_merge}
    \Description{An example illustrating SLD-Merge. The tree is split at node e into two trees (the left and right sides of the dashed line) which share no edges, and only share the vertex e. The SLD-Merge routine merges the two spines formed by the lowest-rank edge
    incident to e in both trees.}
    \vspace{-1.5em}
\end{figure}

\subsection{Merging Dendrograms}\label{sec:merging_dendrograms}
The key component in our merge-based framework is the subroutine called \sldmerge, which can merge the SLDs of two trees under the following conditions:
\begin{itemize}[topsep=0pt,itemsep=0pt,parsep=0pt,leftmargin=15pt]
    \item the trees share exactly one vertex (denoted by $v$),
    \item the trees share no edges.
\end{itemize}
Given these conditions, observe that the union of the two trees will also be a tree.
More abstractly, from the divide-and-conquer viewpoint merging is algorithmically useful if for example, an input tree is split at vertex $v$ and the SLDs of the two resultant trees are computed recursively, and we wish to merge the two SLDs to compute the SLD of the entire tree; see Figure~\ref{fig:sld_merge} for an example.

More formally, let $G_1(V_1,E_1)$ and $G_2(V_2,E_2)$ denote the two input trees with $V_1\cap V_2 = \{v\}$ and $E_1\cap E_2=\varnothing$, and let $D_1$ and $D_2$ denote their SLDs. We assume that the SLDs are maintained as linked lists, where each node points to its parent node. 
\cref{alg:sld-merge} defines the function $\sldmerge(G_1, G_2, v)$ that, given the two trees and their SLDs, returns the SLD of the merged tree.

\begin{algorithm}
\caption{$\sldmerge(G_1, G_2, v)$}\label{alg:sld-merge}
\DontPrintSemicolon
Let $e_1^\star$ and $e_2^\star$ denote the edges with minimum rank incident to vertex $v$ in $G_1$ and $G_2$, respectively.\;
Merge $D_1$ and $D_2$ along the spines of $e_1^\star$ and $e_2^\star$.\;
\end{algorithm}

Given the linked list representation, the spine of a node $e$ is the partial linked list starting at $e$ and ending at the 
root. The nodes have ranks in increasing order from $e$ to the root. Thus, we can apply the standard list merge algorithm
for merging the two (sorted) lists.
The idea here is inspired by the Cartesian tree algorithm of \cite{SB14cartesian}. 
Indeed, it is not hard to see that the Cartesian tree problem on lists is equivalent to single-linkage clustering on path graphs. 
In \cite{SB14cartesian}, the authors employ an elegant divide-and-conquer approach where they split the input list into two halves, compute the Cartesian tree of each half, and then merge them. 

The key idea when merging the two Cartesian trees was that the merge impacts only nodes on certain spines, while the rest of the nodes are ``protected''. 
Interestingly, we prove that such a property holds even when we merge the dendrograms of trees with arbitrary arity. 
In particular, the merge potentially impacts only nodes on the spines of $e_1^\star$ and $e_2^\star$ (defined above). We will call these min-rank edges $e_1^\star$ and $e_2^\star$ as the \defn{characteristic edges} of the merge, and their spines in their respective SLDs as the \defn{characteristic spines}. Therefore, in other words, the parent of nodes that are \emph{not} on the corresponding characteristic spines remains unchanged. In the case where one of the trees is just a single vertex, we do not have a characteristic edge. Here, we call the empty list as the characteristic spine of this tree.

Before proving the correctness of \sldmerge, we first state some useful structural properties of SLDs.

\begin{figure}[t]
    \centering
    \vspace{-1em}
    \includegraphics[width=0.45\textwidth]{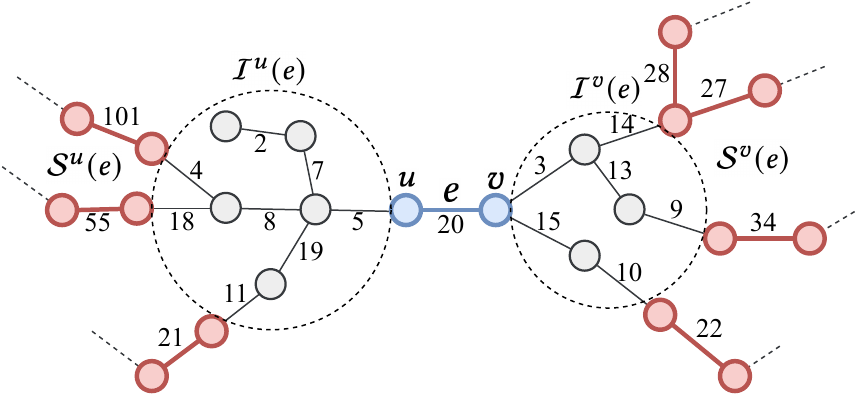}
    \caption{\small Adjacent Superiors and Inferiors (see~\cref{defn:adjsupinf}).}
    \label{fig:adj_sup_inf}
    \Description{A figure illustrating the adjacent superiors and inferiors of an edge.}
    \vspace{-1.5em}
\end{figure}


\begin{definition}[Adjacent Superior and Inferior] \label{defn:adjsupinf}
    Given a tree $G(V,E)$ and an edge $e=(u,v) \in E$, the edge $f \in E$ is an \defn{adjacent superior} of $e$ if $r_e < r_f$ and for every edge $g$ in the unique path between $e$ and $f$, $r_g < r_e$. The edge $f$ is an \defn{adjacent inferior} of $e$ if $r_e > r_f$ and for every edge $g$ in the unique path between $e$ and $f$, $r_g < r_e$. Let $\mathcal{I}^u(e)$ (and $\mathcal{S}^u(e)$) denote the set of adjacent inferior (superior) edges to $e$ that are closer to vertex $u$ compared to vertex $v$. Similarly, we define $\mathcal{I}^v(e)$ and $\mathcal{S}^v(e)$. Let $\mathcal{I}(e) = \mathcal{I}^u(e) \cup \mathcal{I}^v(e)$ and $\mathcal{S}(e) = \mathcal{S}^u(e) \cup \mathcal{S}^v(e)$. See \Cref{fig:adj_sup_inf} for an illustration.
\end{definition}

Observe that the subgraphs induced on the sets of edges $\mathcal{I}^u(e)$ and $\mathcal{I}^v(e)$ will be connected, respectively. 
We have the following lemma about the correspondence of these sets in the output SLD.

\begin{lemma}\label{lem:sld_subtree}
    Let $D$ be the output SLD of the tree $G(V,E)$. For the edge $e=(u,v)\in E$, let $D(e)$ denote the subtree rooted at node $e$, and let $D^u(e)$ denote the subtree rooted at the child of node $e$ that contains the vertex $u$ as a leaf. Similarly, we define $D^v(e)$. Then, $D^u(e) = \mathcal{I}^u(e)$ and $D^v(e) = \mathcal{I}^v(e)$.
\end{lemma}%
\begin{proof}
Observe that during single linkage clustering the edges in $\mathcal{I}(e)$ are processed before edge $e$. Just before $e$ is processed, the clusters on the endpoints are exactly the sets $\mathcal{I}^u(e)$ and $\mathcal{I}^u(e)$. To see this, observe that after $\mathcal{I}(e)$ is processed, the minimum rank edge incident on clusters of $u$ and $v$ is the edge $e$, since all the other incident edges will be the set $\mathcal{S}(e)$. 
Hence, $D^u(e) = \mathcal{I}^u(e)$ and $D^v(e) = \mathcal{I}^v(e)$.
\end{proof}

We also have the following simple and useful observation.
\begin{lemma}\label{lem:star_spine}
    Let $e_1, e_2, \ldots, e_d$ be the set of edges incident to some vertex $u$, sorted by rank. Then, $e_i \in \spine{}(e_1)$, for all $2 \le i \le d$.
\end{lemma}
The proof follows due to the fact that these edges share an endpoint.
Returning to the merge subroutine, we define a node (in $D_1$ and $D_2$) as \defn{protected under the merge} if its parent doesn't change in the output after the merge. We now prove a crucial lemma. Henceforth, when we say ``nearest edge'', the distance here is the unweighted \emph{hop} distance.
\begin{lemma}\label{lem:spine_only_sld}
    Every node in $D_1 \cup D_2$ that is not present in $\spine{}(e_1^\star) \cup \spine{}(e_2^\star)$ is protected under the merge.
\end{lemma}
\begin{proof}
Let $e=(u,v) \in E_1$ such that $e \notin \spine{}(e_1^\star)$. Observe that if $\mathcal{I}(e)$ doesn't change when we union the two trees, then the subtree rooted at $e$ in $D_1$ is protected. This follows from the fact that the subtree rooted at $e$ in $D_1$ corresponds to a (connected) subtree in $G_1$ (by \Cref{lem:sld_subtree}), and SLD is computed correctly on this subtree (by induction). Thus, it is enough to show that $\mathcal{I}(e)$ doesn't change for every $e \notin \spine{}(e_1^\star)$.

Let $a$ be the lowest common ancestor of $e$ and $e_1^\star$ in $D_1$. 
Then, observe that $r_e < r_a$ and $a$ lies on the unique path between $e$ and $e_1'$, where $e_1'$ is its nearest edge incident on $v$.
We also know that $e_1' \in \spine{}(e_1^\star)$ by \Cref{lem:star_spine}.
Therefore, $\mathcal{S}(e)$ cannot change, since any possible change would be via $e_1'$, i.e. via vertex $v$, but this is blocked by edge $a$. Hence, $\mathcal{I}(e)$ cannot change as well, completing the proof. A symmetric argument can be made for edges in $E_2$.
\end{proof}

Finally, we now prove the correctness of \sldmerge.
\begin{theorem}
    $\sldmerge(G_1, G_2, v)$ outputs the correct SLD of $G_1 \cup G_2$.
\end{theorem}
\begin{proof}
From \Cref{lem:spine_only_sld}, we know that nodes not on $\spine{}(e_1^\star) \cup \spine{}(e_2^\star)$ are protected. Consider some edge $e \in \spine{}(e_1^\star)$, and let $e_1'$ denote its nearest edge incident to $v$ in $G_1$. Observe that $e$ doesn't have an adjacent superior on the path from $e$ to $e_1'$. Therefore, after the merge, $e$ might have new adjacent superiors introduced along this path. The parent of $e$ (say $p(e)$) will change iff $r_{p(e)} > r_f$ for some new adjacent superior $f$. We claim that all the new adjacent superiors for $e$ belong to $\spine{}(e_2^\star)$.

To see this, consider some new adjacent superior $f$, and let $e_2'$ denote its nearest edge incident to $v$ in $G_2$. 
Then, for all edges $g$ in the unique path between $e_2'$ and $f$, $r_g < r_e$ (by definition), which implies $r_g < r_f$. Thus, $f$ too doesn't have an adjacent superior on the path from $f$ to $e_2'$. From the proof of \Cref{lem:spine_only_sld}, $f$ is not protected, implying that $f \in \spine{}(e_2^\star)$.

Since $\spine{}(e_2^\star)$ is sorted, if $p(e)$ changes, it will be the first node $f$ in the list with rank greater than $r_e$. Thus, \sldmerge is correct.
\end{proof}

\subsection{Optimal Algorithm via Tree Contraction}\label{sec:tree_contraction}
We now describe an optimal $O(n\log h)$ work and $O(\log^2n \log^2 h)$ depth algorithm for computing SLDs; we will refer to this algorithm as \sldtreecon. We achieve this with the help of the merge subroutine described above and parallel tree contraction. Our bounds match those of the following comparison-based lower bound stated next, which we show in \cref{sec:lb}:

\begin{restatable}{lemma}{lowerbound}\label{thm:lb}
For any $\lfloor \log n \rfloor \leq h \leq n-1$, there is an input that every comparison-based SLD algorithm requires $\Omega(n\log h)$ work to compute the parent of every edge in the output dendrogram.
\end{restatable}

As indicated previously, with the help of \sldmerge, we can design divide-and-conquer algorithms that partition the input tree into smaller subtrees, compute their respective SLDs in recursive rounds, and finally applies the \sldmerge subroutine, suitably, to obtain the overall SLD. 
A critical task here is to structure these recursive rounds as efficiently as possible, with low depth; parallel tree contraction~\cite{reif94treecontraction} provides one such structure. 

As discussed earlier (in \Cref{sec:prelims}), parallel tree contraction defines a convenient low-depth hierarchical decomposition (or clustering) of trees. Our core idea is to maintain the SLD of each cluster (which is a connected subtree) and apply \sldmerge appropriately during rakes and compresses to obtain the SLD of the merged cluster. 
This way, by the end of tree contraction when we have a single cluster containing the entire input tree, we will have constructed its corresponding SLD, as required. 
First, we will discuss how rakes and compresses can be realized as a couple of \sldmerge operations, assuming merges are performed as standard (linked) list merges. However, this approach leads to sub-optimal work and depth bounds. 
Second, we will discuss how to optimize the merges by additionally maintaining certain spines in a more efficient data structure, and prove optimal work and depth bounds.

\subsubsection{A Sub-optimal Tree-Contraction Algorithm.}\label{sec:subopt}\ 
Formally, for a cluster represented by $u$ during tree contraction, let $G_u$ denote the induced subtree on the vertices in cluster $u$, and let $D_u$ denote the SLD of $G_u$. 
Consider some \rake{} operation given by $\rake(u,v)$, raking vertex $v$ into $u$. 
The subtrees $G_u$ and $G_v$ are connected via the edge $e=(u,v)$. For convenience, let $G_{uv}$ denote the union of $G_u$ and $G_v$, i.e. the cluster obtained after performing the rake. 
We can implement the rake using two steps: (1) Add the edge $e=(u,v)$ to $G_v$ to obtain the subtree $G_v'$, and (2) merge the subtrees $G_u$ and $G_v'$. 
We compute the SLD of $G_{uv}$ in the following two step process:
\begin{algorithm}
\caption{$\rake(u,v)$}\label{alg:rake}
\DontPrintSemicolon
$G_v' \gets G_v \cup \{e\}$, and\newline 
$D_v' \gets \sldmerge(G_v, \{e\}, v)$,\;\label{line:rake_step1}
$G_{uv} \gets G_u \cup G_v'$, and\newline
$D_{uv} \gets \sldmerge(G_u, G_v', u)$.\;\label{line:rake_step2}
\end{algorithm}
\begin{figure}[!htpb]
    \centering
    \includegraphics[width=0.35\textwidth]{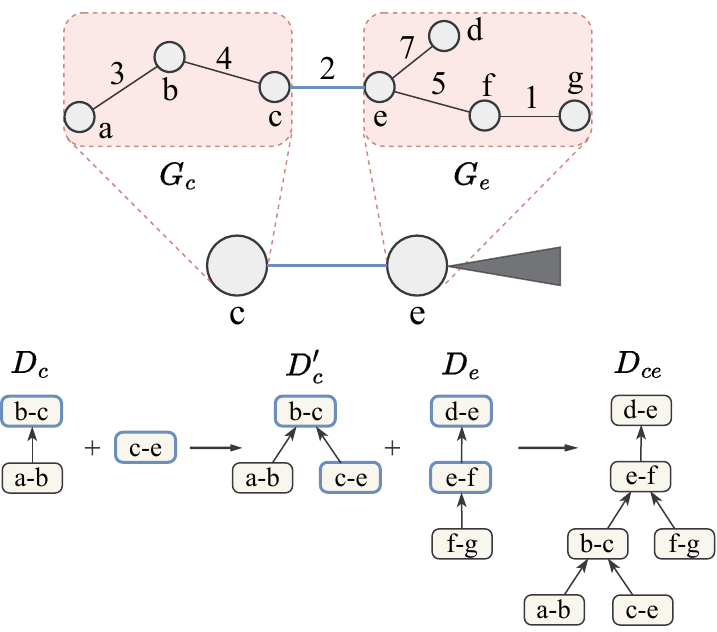}
    \caption{\small An example illustrating the two-step rake (see \Cref{alg:rake}). Here, we perform $\rake(e,c)$, which rakes $c$ into $e$.}\label{fig:rake-compress}
    \Description{An example illustrating the two-step rake (see Algorithm rake(u,v)). Here, we perform rake(e,c), which rakes c into e.}
    \vspace{-1.25em}
\end{figure}

Compress can be performed in an identical fashion: given an operation $\compress(u,v,w)$, we can choose to merge $v$ with $u$ (arbitrarily) and a similar viewpoint, as in rake, can be applied for this merge. Therefore, each rake and compress operation (identically) performs two \sldmerge operations. 

During tree contraction, if we execute the rake/compress operations in parallel, we could encounter race conditions. 
For instance, multiple clusters might get raked or compressed into the same cluster, and we make no assumptions about the structure of the input trees (e.g., some applications of tree contraction assume bounded-degree input trees). 
Note that these are the only race conditions that occur in any rake/compress round of tree contraction.

We can handle parallel rake or compress operations as follows: let $v_1, v_2,\ldots, v_d$ be vertices that are being raked (or compressed) into the same cluster $u$. 
Observe that the step~(1) described above doesn't affect $G_u$ or $D_u$. 
Thus, this step can be run safely in parallel. Let $G_u''$ denote the tree formed by taking the union of trees $G_{v_1}',G_{v_2}', \ldots, G_{v_d}'$. Our idea is to first compute the dendrogram of $G_u''$, and then finally run $\sldmerge(G_u,G_u'',u)$.
We can compute the dendrogram of $G_u''$ by simply merging the dendrograms of $G_{v_1}',G_{v_2}', \ldots, G_{v_d}'$ together, i.e. merging the $d$ sorted lists given be $\spine{}((u,v_i))$ for each $i \in [d]$. This is correct due to a simple extension of \Cref{lem:star_spine}: since all these edges incident to $u$ will be on the same spine, their respective spines in $G_{v_i}'$ will also end up on the same spine. This can be computed quickly by running a parallel reduce operation with \sldmerge as the reduce function, resulting in depth $O(\log d)$ ($= O(\log h)$ by \Cref{lem:star_spine}) times the depth of \sldmerge for all the rakes (and compresses) on $u$. 

The naive (sequential) linked list based implementation of \sldmerge has $O(h)$ work and depth. Thus, if we charge each merge cost of $O(h)$ to the vertex that is being raked/compressed, we get an overall work bound of $O(nh)$ for this sub-optimal version of \sldtreecon. As discussed, each rake/compress round will have a worst case depth of $O(h\log h)$. Since the number of such rounds is $O(\log n)$, the overall depth will be $O(h\log h\log n)$. We will now prove its correctness.

\begin{lemma}\label{lem:subopt_sldtreecon}
    The sub-optimal version of \sldtreecon correctly computes the SLD.
\end{lemma}
\begin{proof}
    Observe that both step~\ref{line:rake_step1} and step~\ref{line:rake_step2} are merges between subtrees that satisfy the requirement mentioned in \Cref{sec:merging_dendrograms}, i.e. the subtrees share exactly one vertex and no edges. Thus, rake correctly computes the dendrogram of the merged cluster; a similar argument works for compress. Since tree contraction is just a sequence of rake and compress operations, the correctness follows by a simple inductive argument.
\end{proof}

\subsubsection{Optimizing the merge step.}\ 
We now describe how to optimize the merge step. From \Cref{sec:merging_dendrograms}, we know that merging affects only nodes on the characteristic spines associated to that merge. Our main idea is that, in addition to the linked list representation for storing the output of the dendrogram, for each cluster we (try to) maintain the characteristic spines, corresponding to the next (future) merge involving that cluster, in \emph{parallel binomial heaps} and perform merges via these heaps. We chose binomial heaps since (to the best of our knowledge) it is the only data structure that supports fast merge (or meld) and can support low-depth parallel filter operations. We describe the heap interface and the cost bounds in \Cref{sec:heaps}.
As we will see, if we perform only rakes, it is easy to always store these characteristic spines. 
However, compress operations pose a significant challenge towards exactly storing the characteristic spines. 
Nevertheless, if compresses are performed carefully (in terms of which cluster to merge with), we show that the spine consequently stored at each cluster is always sufficient for every merge operation performed in the future. 
Henceforth, when we mention heaps, we refer to parallel binomial heaps.

Extending the notation from before, let $H_u$ denote the min-heap associated with the cluster represented by $u$.
We first extend the notion of protected nodes defined in \Cref{sec:merging_dendrograms} as follows: a node $e$ is \defn{protected} if its parent node is identical to its parent in the final output. We would like to maintain the following invariant: (1) nodes present in the heap are potentially \emph{not} protected and correspond to some spine in the dendrogram associated to that cluster, and, (2) all nodes in that cluster \emph{not} present in the heap are protected. 
We also show that when a node is deleted from its heap during the course of the algorithm, it is definitely protected. 
Thus, we ensure that we update the output (i.e. parent array) only when a node is deleted from its heap. 
This invariant helps us carefully charge the associated merge costs to these nodes.

We will now describe optimized versions of the previously described two-step rake and compress operations, in which \sldmerge will be implemented in a white-box manner via the heaps.

\begin{figure*}[ht]
\centering
\renewcommand{\arraystretch}{2}
\vspace{-1em}
\begin{tcolorbox}[colback=white, enhanced, fontupper=\small, fontlower=\tiny, width=0.7\linewidth, tabularx={>{\centering\arraybackslash}l|>{\centering\arraybackslash}c|>{\centering\arraybackslash}X|>{\centering\arraybackslash}X}]
 & & Heaps & Output\\\hline
\textbf{Init:} & \makecell{\includegraphics[width=0.4\linewidth]{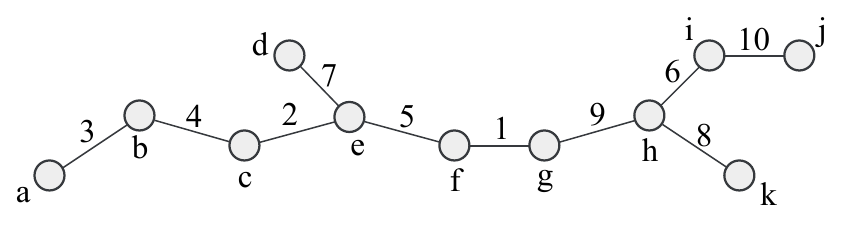}} & & \\\hline
\makecell[l]{\textbf{Round 1:}\\$\rake(a,b)$\\$\rake(d,e)$\\$\rake(h,k)$\\$\rake(i,j)$} & \makecell{\includegraphics[width=0.4\linewidth]{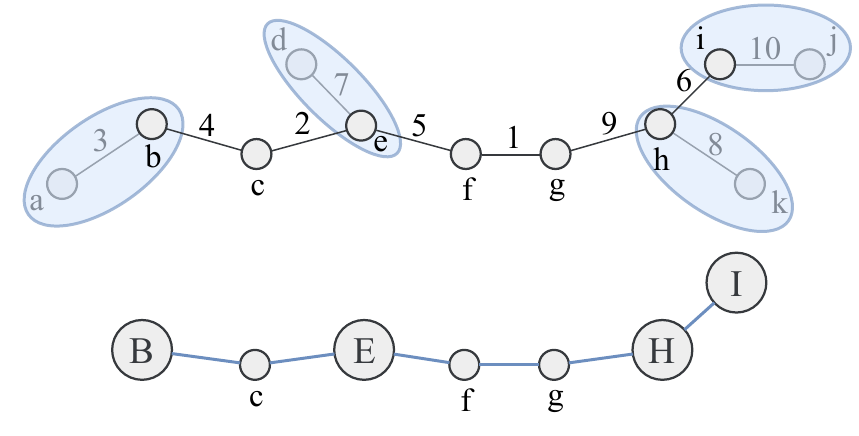}} & \makecell{\includegraphics[width=\linewidth]{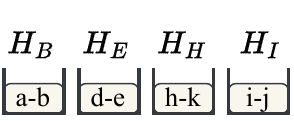}} & \makecell{\includegraphics[width=\linewidth]{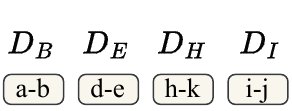}}\\\hline
\makecell[l]{\textbf{Round 2:}\\$\compress(B,c,E)$\\$\compress(E,f,g)$\\$\compress(g,H,I)$} & \makecell{\includegraphics[width=0.4\linewidth]{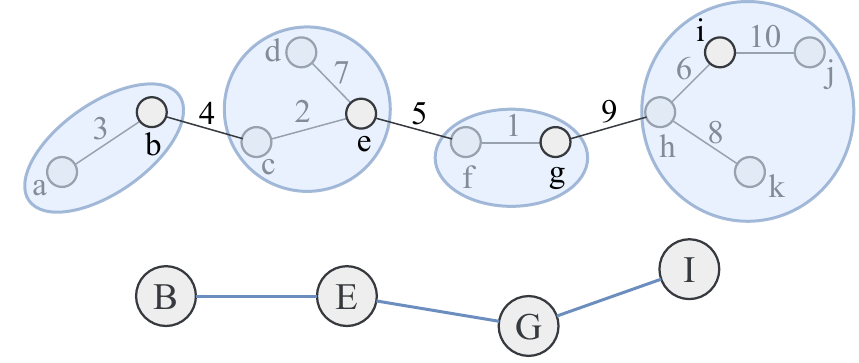}} & \makecell{\includegraphics[width=\linewidth]{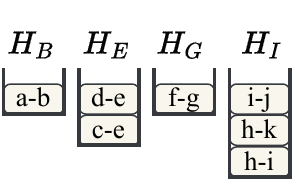}} & \makecell{\includegraphics[width=\linewidth]{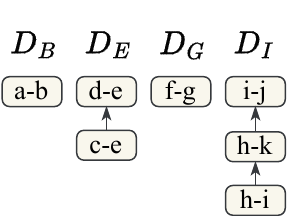}}\\\hline
\makecell[l]{\textbf{Round 3:}\\$\rake(B,E)$\\$\rake(G,I)$} & \makecell{\includegraphics[width=0.4\linewidth]{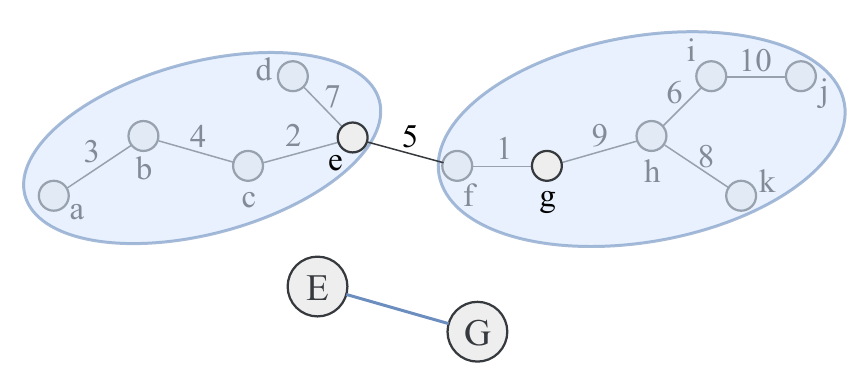}} & \makecell{\includegraphics[width=0.7\linewidth]{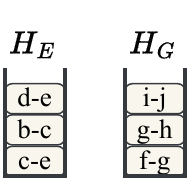}} & \makecell{\includegraphics[width=\linewidth]{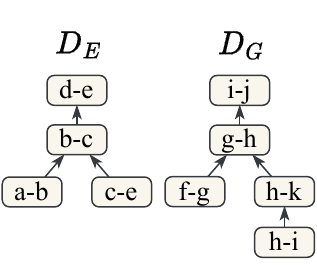}}\\\hline
\makecell[l]{\textbf{Round 4:}\\$\rake(E,G)$} & \makecell{\includegraphics[width=0.4\linewidth]{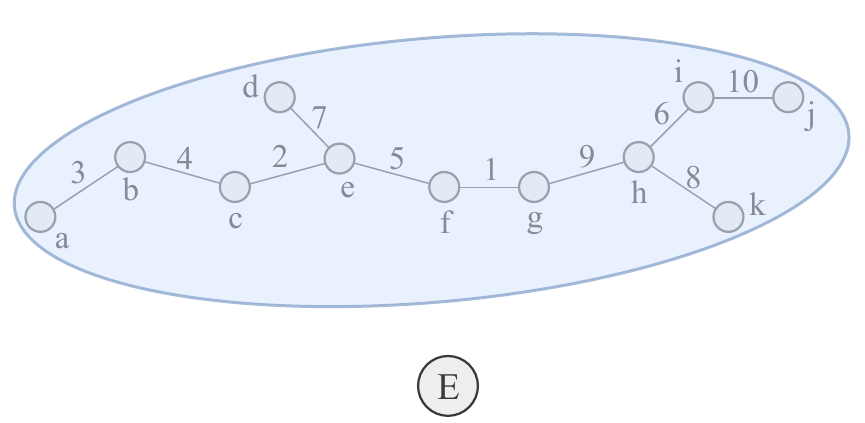}} & \makecell{\includegraphics[width=0.27\linewidth]{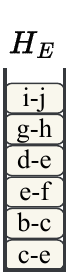}} & \makecell{\includegraphics[width=\linewidth]{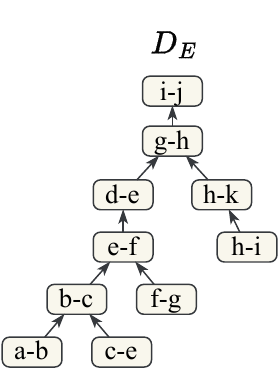}}
\end{tcolorbox}
\caption{\small A full example of \sldtreecon: the first column represents the rakes/compresses performed in that round; the second column displays the clustering obtained by tree contraction (as well as a compact representation); the third column displays the (non-empty) heaps maintained at each cluster; and the fourth column represents the (non-empty) SLDs of each cluster.}
\vspace{-2.25em}
\label{fig:full_example}
\Description{A full example of SLD-TreeContraction: the first column represents the rakes/compresses performed in that round; the second column displays the clustering obtained by tree contraction (as well as a compact representation); the third column displays the (non-empty) heaps maintained at each cluster; and the fourth column represents the (non-empty) SLDs of each cluster.}
\end{figure*}

\myparagraph{\boldmath$\rake(u,v)$} 
Rake is implemented as follows:

\begin{algorithm}
\caption{$\rake(u,v)$ (optimized version)}\label{alg:rake_optimized}
\SetKwFor{parfor}{parfor}{do}{}%
\AlgoDisplayBlockMarkers\SetAlgoBlockMarkers{}{end}%
\DontPrintSemicolon
Let $e = (u,v)$.\;
$S, H_v' = H_v.\textsc{filter\_and\_insert}(e)$.\; \label{rakeopt:filter}
$H_{u} \gets \textsc{Meld}(H_u, H_v')$.\;
// Update output; $S$ contains edges $f \in H_v$ that had $r_f < r_e$ and were filtered out on \cref{rakeopt:filter}.\;
\If{$S.\mb{size} > 0$}{
    Sort $S$ according to rank.\;
    \parfor{each $i=1$ to $S.\mb{size}-1$}{
        $p(S[i]) = S[i+1]$\;
    }
    $p(S[S.\mb{size}]) = e$.\;
}
\end{algorithm}
Note that vertex $u$ will be the representative of the merged cluster, hence the new spine is stored in $H_u$.
Now, we prove the following claim about the set $S$.
\begin{claim}\label{claim:protect_rake}
    The nodes in the set $S$ computed during rake are protected after the merge.
\end{claim}
\begin{proof}
    To see this, observe that for each $f\in S$, $r_f < r_e$. Thus, $e$ is either an adjacent superior to $f$, or $f$ has an adjacent superior on the unique path between $f$ and $e$. Since cluster $v$ is being raked, for any future merge corresponding to some edge $g$, the unique path between $f$ and $g$ will contain $e$. Thus, $e$ protects $f$ from all future merges. Further, since the nodes in $S$ are on the same spine, they will be present in sorted order, and $e$ will be the parent of the max-rank edge in $S$. 
\end{proof}

\myparagraph{\boldmath$\compress(u,v,w)$} Compress is executed in a similar manner as rake, except for one difference: the cluster $v$ will always merge with the cluster along the {\em lesser rank edge} (previously we could merge with an arbitrary neighboring cluster). The pseudocode is shown in \cref{alg:compress_optimized}.
Next, we prove a similar claim, as in rake, about the set $S$.

\begin{algorithm}
\caption{$\compress(u,v,w)$ (optimized version)}\label{alg:compress_optimized}
\SetKwFor{parfor}{parfor}{do}{}%
\AlgoDisplayBlockMarkers\SetAlgoBlockMarkers{}{end}%
\DontPrintSemicolon
Let $e_1=(u,v)$ and $e_2=(v,w)$.\;
\If {$r_{e_1} > r_{e_2}$}{
    swap $u$ and $w$ // Thus, we will have $r_{e_1} < r_{e_2}$.\;
}
$S, H_v' = H_v.\textsc{filter\_and\_insert}(e_1)$.\; \label{compressopt:filter}
$H_{u} \gets \textsc{Meld}(H_u, H_v')$.\;
// Update output; $S$ contains edges $f \in H_v$ that had $r_f < r_{e_1}$ and were filtered out on \cref{compressopt:filter}.\;
\If{$S.\mb{size} > 0$}{
    Sort $S$ according to ranks.\;
    \parfor{each $i=1$ to $S.\mb{size}-1$}{
        $p(S[i]) = S[i+1]$\;
    }
    $p(S[S.\mb{size}]) = e_1$.\;
}
\end{algorithm}

\begin{claim}\label{claim:protect_compress}
    The nodes in the set $S$ computed during compress are protected after the merge.
\end{claim}
\begin{proof}
    Consider some edge $f \in S$. We have $r_f < r_{e_1} < r_{e_2}$. Let $e_1'$ and $e_2'$ denote the adjacent superiors on the paths from $f$ to $e_1$ and $e_2$, respectively. Consider some future merge involving an edge $g$.

(1) If $e_1' \ne e_1$ and $e_2' \ne e_2$, then $e_1'$ and $e_2'$ protect $f$ from any future merges. (Interestingly, in this case, $f$ would have already been protected by a previous rake/compress operation involving either $e_1'$ or $e_2'$, but this is not important for this proof.)

(2) If $e_2' \ne e_2$, then $f$ is protected by $e_2'$ in the case when $g$ is closer to $e_2$ than $e_1$. Similarly, $f$ is protected by $e_1$ in the case $g$ is closer to $e_1$ than $e_2$.

        (3) Finally, we have the case $e_2' = e_2$. If $g$ is closer to $e_1$, $e_1'$ protects $f$. Consider the case when $g$ is closer to $e_2$ than $e_1$. Observe that $g$ cannot have an endpoint in the cluster containing $f$ until $g=e_2$. After merging along $e_2$, any future merge corresponding to edge $g$ that is closer to $e_2$ will not affect $f$ since it will be protected by $e_2$. We claim that the merge corresponding to $e_2$ also doesn't affect $f$.
        
        Firstly, if $r_{e_1'} < r_{e_2}$, then we are done since both $e_1'$ and $e_2$ are adjacent superiors to $f$ but $e_1'$ will be processed before $e_2$ in single-linkage clustering. Hence, $e_1'$ will be the parent of $f$ in the output SLD. Now, assume $r_{e_1'} > r_{e_2}$. We will prove that this is not possible. We know that $e_1'$ and $f$ are present on the path between $e_1$ and $e_2$. Consider the merge corresponding to the edge $e_1'$. Observe that this has to be a compress (none of its endpoints can have degree 1). Since it is a compress, the counterpart edge, say $e_1''$, will have rank greater than $e_1'$. We know that $e_1''$ has to be present on the path between $e_1$ and $e_2$. It cannot be present between $f$ and $e_1'$ (contradiction that $e_1'$ is an adjacent superior). Further, it cannot be present between $f$ and $e_2$ as well (contradiction that $e_2$ is an adjacent superior). Thus, it must be present between $e_1$ and $e_1'$. However, the present compress corresponding to $e_1$ cannot be performed until the merge corresponding to $e_1''$ is performed. If we repeat the same argument for $e_1''$, we will reach the conclusion that either $r_{e_1'} < r_{e_2}$, or $r_{e_1'} < r_{e_1}$, both leading to contradictions.
\end{proof}

The following correspondence is true of \cref{alg:rake_optimized}, barring for the update output step: the filter and insert step corresponds to step~\ref{line:rake_step1}, whereas the meld step corresponds to step~\ref{line:rake_step2}, respectively in \Cref{alg:rake}, and similarly for the optimized version of compress. 
The main difference is that the optimized versions of rake and compress essentially delay the update to the output until the nodes get protected, thus having to update the parent of any node at most once. 

This correspondence helps us handle the race conditions mentioned before, i.e. multiple clusters getting raked/compressed into the same cluster in the same round. 
Using the same notation as before, we perform the filter and insert step at the clusters being raked/compressed, in parallel, to obtain the heaps $H'_{v_1}, H'_{v_2},\ldots, H'_{v_d}$. 
Then, we merge all of these heaps together to obtain the heap $H''_{u}$ by running a parallel reduce operation (with \emph{meld} now as the function). 
Finally, we meld $H_u$ and $H_u''$ to obtain the final spine. 

For performing the merges correctly, we need to ensure that we indeed merge the characteristic spines associated to that merge, as required by \sldmerge. 
Instead, we show that for every merge performed, at least one of the spines will be the exact characteristic spine, and the other spine will be sufficient for the corresponding merge, as stated by the following lemma.

\begin{restatable}{lemma}{sldtreeconheaps}
\label{lem:sldtreecon_heaps}
For any cluster $u$ during tree contraction, $H_u$ stores a spine in the SLD of $G_u$ satisfying one of the following properties:
\begin{enumerate}
    \item $H_u$ contains the characteristic spine corresponding to the next merge involving $u$.
    \item Let $e$ be the characteristic edge in $u$, and $f$ be the characteristic edge in the other subtree corresponding to the next merge involving cluster $u$. Then, $H_u$ contains $\spine{}(e')$ such that $\spine{}(e') \subseteq \spine{}(e)$ and $r_{e'} < r_f$.
\end{enumerate}
\end{restatable}
\begin{proof}
    We prove the lemma by induction on the tree contraction rounds. 
    Initially, all the clusters are singleton and the heaps are empty, which corresponds to the characteristic spine of the next merge. 
    Let us assume that the clusters at the end of round $k-1$ store spines in their heaps that satisfy one of the properties stated. 
    Now, consider some cluster $u$. 
    At round $k$, if no clusters merge with $u$, we are done. 
    
    Suppose that in round $k$, the clusters $v_1, v_2, \ldots, v_d$ get \emph{raked} into $u$. 
    We first look at the filter and insert step: the single edge $e_i=(u,v_i)$ corresponds to the characteristic spine for its subtree. 
    By induction, in the first case $H_{v_i}$ contains the characteristic spine, in which case the merge is correct and we obtain the $\spine{}(e_i)$ in $H_{v_i}'$. Otherwise, let $\spine{}(f_i)$ be the required characteristic spine. 
    Since we are in the second case, we know that $H_{v_i}$ contains $\spine{}(f_i') \subseteq \spine{}(f_i)$ such that $r_{f_i'} < r_{e_i}$. 
    Observe that the output of merging $\spine{}(e_i)$ and $\spine{}(f_i)$ is equivalent to the output of merging $\spine{}(e_i)$ and $\spine{}(f_i')$, since $\spine{}(f_i') \subseteq \spine{}(f_i)$. 
    Thus, the merge is correct, and we similarly obtain $\spine{}(e_i)$ in $H_{v_i}'$. 
    By \Cref{lem:star_spine}, the characteristic spine of the SLD of $G_u''$ when merging with the SLD of $G_u$ is nothing than the union of $\spine{}(e_i)$ over all $i=1,\ldots,d$. 
    Thus, $H_u''$ stores the characteristic spine of the next merge. 
    By induction, if $H_u$ stores the characteristic spine, the merge is performed correctly. 
    Otherwise too, with a similar (equivalence) argument as before, the merge is correct.

    Now, let us instead consider that round $k$ performs compresses. Let $e_{i_1}=(u,v_i)$ and $e_{i_2}=(v_i,w_i)$, such that $r_{e_{i_1}} < r_{e_{i_2}}$. The same argument, as in the case of rakes, can be extended to work here. Now, we are left to prove that the final computed spine in $H_u$ satisfies one of the stated properties.

    In the future, $u$ can be involved in a merge along the edge $e_{i_2}$. However, the final spine computed in $H_u$ contains $\spine{}(e_{i_1})$, which is a subset of the characteristic spine corresponding to the merge along $e_{i_2}$; the characteristic edge would be the min-rank edge incident at that vertex, but by \Cref{lem:star_spine}, $\spine(e_{i_1})$ will be a subset of the characteristic spine. Thus, if $e_{i_2}$ is the next merge, $H_u$ satisfies property~(2). 
    
    Since we do not delete nodes from $H_u$ until cluster $u$ is raked/compressed, it satisfies property~(2) until then. By induction, $H_u$ will continue to satisfy property~(2) for all other possible future merges determined by compresses (involving $u$) performed in the past. Hence, it will continue to satisfy this property even after round $k$ (this holds even in the case when the round performs rakes).

    Thus, by induction, the heaps stored at any cluster always satisfies one of the two stated properties.
\end{proof}

At the end of tree contraction when we obtain a single cluster, observe that we have a spine stored in the heap associated to that cluster, whose parent nodes are not updated in the output. Since there are no more merges left and the nodes form a spine, we can just sort all the nodes in the heaps and assign parents, identically to the update output step described in \Cref{alg:rake_optimized,alg:compress_optimized}. 
With the algorithm fully described, we will now prove the correctness and work-depth bounds of \sldtreecon.
\begin{theorem}\label{thm:sldtreecon}
    \sldtreecon correctly computes the SLD, and runs in $O(n\log h)$ work and $O(\log^2n\log^2 h)$ depth.
\end{theorem}
\begin{proof}
    The proof idea is similar to \Cref{lem:subopt_sldtreecon}, except the fact that \sldmerge is executed in a white-box manner. By \Cref{lem:sldtreecon_heaps}, every cluster maintains the correct characteristic spines (or a sufficient version of it) in their heaps, and thus rakes and compresses are executed correctly. By \Cref{claim:protect_rake} and \Cref{claim:protect_compress}, the output will be updated correctly. Thus, as tree contraction is a sequence of rakes and compresses, we can apply a simple inductive argument as before to complete the correctness argument.

    We will now analyze the work and depth bounds. 
    Recall that we decided to use parallel binomial heaps since they support both fast meld ($O(\log (s)$ work and depth), where $s$ is the size of the merged heap) and low-depth parallel filter operations ($O(k\log s)$ work and $O(\log^2 s)$ depth, where $k$ is the number of nodes filtered). 
    Note that since at any point the nodes in any heap correspond to some spine, the size of every heap will always be $O(h)$.

    \myparagraph{Work Analysis} 
    If $k$ nodes are protected at any rake/compress, the total work for the filter step will be $O(k\log h)$. 
    The subsequent output SLD update step also performs $O(k\log h)$ work (sorting plus update). 
    Thus, we can charge $O(\log h)$ work to each of the protected nodes. 
    Since each node is protected at most once, the overall work incurred will be $O(n\log h)$. 
    Next, each meld operation requires $O(\log h)$ work. 
    Since every rake/compress operation is associated to exactly one edge of the input tree, we can associate this $O(\log h)$ work to that edge, leading to a total of $O(n\log h)$ work across all edges. 
    Hence, the overall work of the algorithm is $O(n\log h)$.
    
    \myparagraph{Depth Analysis}
    Let us now analyze the depth of the algorithm. The number of rake and compress rounds is at most $O(\log n)$. The cost of spawning threads within a round is $O(\log n)$. The depth of the heap filter step is $O(\log^2h)$, the depth of update output and heap meld is $O(\log h)$ (times $O(\log h)$ for meld due to the \emph{reduce}). Therefore, the overall depth of each rake/compress round is $O(\log n\log^2h)$, giving the overall depth bound of $O(\log^2n\log^2h)$.
\end{proof}
\section{Practical Algorithms}
In this section, we describe two algorithms for dendrogram computation, both of which have strong provable guarantees (both achieve the optimal work bounds we showed in \cref{sec:merging_dendrograms}) but are also implementable and achieve good practical performance.
The first is an \emph{activation-based} algorithm that achieves the optimal work bound of $O(n\log h)$ and has $O(h\log n)$ depth (\cref{sec:paruf}).
The second algorithm is a twist on the tree contraction algorithm described in \Cref{sec:tree_contraction} that first performs tree contraction, and subsequently {\em traces} the tree contraction structure to identify the parent of each edge in the dendrogram (\cref{sec:rctt}). 
It also achieves optimal work, and additionally runs in worst-case poly-logarithmic depth.


\subsection{Activation-Based Algorithm (\paruf{})}\label{sec:paruf}
The sequential Kruskal algorithm processes edges in increasing order of rank, thus emulating the process of single-linkage HAC and building the output dendrogram in a bottom-up fashion, one node at a time. 
We can generalize this approach to safely process multiple edges at the same time by building the dendrogram in a bottom-up fashion, one ``level'' at a time.
This approach is similar to the nearest-neighbor chain algorithm, a well-known technique for HAC \cite{benzecri1982construction} that obtains good parallelism in practice for other linkage criteria such as average-linkage, and complete-linkage~\cite{dhulipala2021hierarchical,sumengen2021scaling, yu2021parchain}.

Indeed, the algorithm we propose can be viewed as an optimized and parallelized version of the sequential algorithm in \cite{dhulipala2021hierarchical}. 
The striking difference between our activation-based algorithm and other nearest-neighbor chain algorithms is that our algorithm is {\em asynchronous} and only requires a single instance of spawning parallelism over the set of edges, whereas all other nearest-neighbor chain implementations we are aware of run in synchronized rounds. 

The following simple but important observation allows us to process multiple edges at a time:
\begin{lemma}[folklore]\label{lem:local_min}
If an edge $e=(u,v)$ is a local minima, i.e. $r_e < r_f$ for all other edges $f$ incident to the clusters containing $u$ and $v$, then the clusters $u$ and $v$ can be safely merged.
\end{lemma}
The proof follows due to the fact that single-linkage clustering processes edges in sorted order of ranks, and hence edge $e$ will be processed before all of the other edges incident to its endpoints.
In other words, the edge $e$ is processed only when it becomes the minimum rank edge incident to the clusters on both of its endpoints. 
We also have the following useful lemma.

\begin{lemma}\label{lem:sld_parent}
Let $e=(u,v)$ be a local minima. Then, the parent of node $e$ in the output SLD will correspond to the minimum rank edge incident on the merged cluster $uv$.
\end{lemma}
\begin{proof}
Once $u$ and $v$ are merged along $e$, let $e_1, e_2, \ldots, e_d$ denote the set of edges incident to the merged cluster $uv$, in sorted order of ranks.
The parent of $e$ is the first edge that merges the cluster $uv$ with some other cluster. 
However, $uv$ can merge only via one of $e_1, e_2, \ldots,$ or $e_d$, and the first edge to be processed among them will be $e_1$, by the definition of single-linkage clustering. 
Therefore, $e_1$ will be the parent of $e$ in the output SLD.
\end{proof}

Based on these observations, we now describe an asynchronous activation-based algorithm that we call \paruf.
We give the pseudocode in \cref{alg:activation_based}.
The idea is natural: when an edge becomes a local minima, we merge it.
We maintain the set of edges incident to a cluster in a meldable min-heap (which we call the \defn{neighbor-heap} of that cluster). 
Each unmerged edge will be present in two neighbor-heaps corresponding to the clusters containing it's endpoints. 
The element at the top of a cluster's heap will correspond to the min-rank edge incident to that cluster. 
We maintain the cluster information using a Union-Find data structure.
Note that due to our strategy of only processing the local minima in parallel, we can use any sequential Union-Find structure with path compression.
To identify if an edge is ready to be processed, i.e. it is a local minima, for each edge $e \in E$ we maintain an integer $status(e)$ value:
\begin{align*}
    \mb{status}(e) = \begin{cases}
        2, & \emph{ready}\\
        1, & \emph{almost ready}\\
        0, & \emph{not ready}\\
        -1, & \emph{inactive}
    \end{cases}
\end{align*}
An edge is {\em ready} if it is at the top of both the neighbor-heaps of its endpoints, i.e. it is a local minima. An edge is {\em almost ready} if it is at the top of the neighbor-heap of only one of it's endpoints. If it is not on top of either neighbor-heaps, the edge is {\em not ready}. Finally, if the edge has already been merged/processed, it is {\em inactive}.

The overall algorithm is given in \Cref{alg:activation_based}.
The updates/accesses to $\mb{status}(e)$ (\cref{{paruf:while}} and \cref{paruf:atomicinc}) must be atomic since it could be updated/accessed by both of it's endpoints simultaneously. 
We 
now show
that the rest of the steps of the algorithm do not have any race conditions, and that the algorithm is efficient:

\begin{algorithm}[ht]
\caption{Activation-based Algorithm (\paruf)}\label{alg:activation_based}
\DontPrintSemicolon
\SetKwFor{parfor}{parfor}{do}{}%
\SetKw{break}{break}%
\AlgoDisplayBlockMarkers\SetAlgoBlockMarkers{}{end}%
$\mb{F} \leftarrow$ Initialize Union-Find with all singleton clusters\; \label{paruf:inituf}
Initialize the output SLD: $\forall e\in E$, $p(e)=e$\; \label{paruf:sldiit}
$\mb{heaps} \leftarrow$ Initialize neighbor heaps\; \label{paruf:initheap}
Initialize $\mb{status}(.)$ values\;\label{paruf:initstatus}


\parfor{each $e \in E$}{\label{paruf:parfor}
    $\mb{cur} \gets e$\;
    \While{$\CAS(\mb{status}(\mb{cur}), 2, -1)$}{\label{paruf:while}
        $(u,v) \gets \mb{cur}$\;
        $(u', v') \gets (\mb{F}.\textsc{Find}(u), \mb{F}.\textsc{Find}(v))$\; \label{paruf:find}
        $w \gets \mb{F}.\textsc{Union}(u', v')$\;\label{paruf:union}
        $\mb{heaps}(u').\textsc{delete\_min}()$\;
        $\mb{heaps}(v').\textsc{delete\_min}()$\; \label{paruf:delmin}
        $\mb{heaps}(w) \gets \textsc{Meld}(\mb{heaps}(u'), \mb{heaps}(v'))$\; \label{paruf:meld}
        \If{$\mb{heaps}(w)$ is empty}{ \label{paruf:break}
            \break\;
        }
        $\mb{new\_cur} \gets \mb{heaps}(w).\textsc{Top}()$\; \label{paruf:heaptop}
        $p(\mb{cur}) = \mb{new\_cur}$\; \label{paruf:setparent}
        $\ati(\mb{status}(\mb{new\_cur}))$\; \label{paruf:atomicinc}
        $\mb{cur} \gets \mb{new\_cur}$\; \label{paruf:updatecur}
    }
}
\end{algorithm}

\begin{restatable}{theorem}{parufwork}
    \paruf correctly computes the SLD, and runs in $O(n\log h)$ work and $O(h\log n)$ depth.
\end{restatable}

\myparagraph{Further Optimizations and Implementation}
As we will see in \Cref{sec:experiments}, the height of the resultant SLD is typically large in many instances. 
However, in most cases, the number of nodes in each level of the output dendrogram, as we go upwards, converges to $1$ quickly. 
In other words, the number of local-minima edges is exactly one the majority of the time, rendering ParUF ineffective.
However, we can apply a very simple optimization in this case. 

If the number of local-minima edges is $1$, this means we have to process the edges one-by-one. 
However, we know that they will be processed in sorted order. Thus, when running ParUF, if the number of local-minima edges (or the number of \emph{ready} edges) drops to $1$, we can stop and compute the set of remaining edges, say $E'$. 
Then, we sort $E'$ based on rank, and assign $\mb{parent}[E'[i]] = E'[i+1]$.
In terms of finding out the number of ready edges, a simple approach is to periodically stop and check the count. This optimization provides incredible speed-ups in most of our experiments. 
However, it is not too hard to generate adversarial inputs that have low parallelism, but elude this strategy (for instance, if the output dendrogram has two nodes in each level for the majority of the time.)

\subsection{RC-Tree Tracing Algorithm (\rctt{})}\label{sec:rctt}
Implementing a fast and practical algorithm that computes the SLD by leveraging the properties of \sldmerge is highly non-trivial.
The practical bottleneck of a faithful implementation of \sldtreecon, our merge-based algorithm appears to be the need to maintain meldable heaps supporting the \emph{heap-filter} operation for merging spines.
In this section, we explore a few more structural properties of \sldmerge and parallel tree contraction to design a fast and practical $O(n\log n)$ work and $O(\log^2n)$ depth algorithm for computing the SLD that completely removes the requirement of maintaining the spines. 
The idea is to use the RC-tree representation of the tree contraction process and apply a post-processing {\em tracing} step to compute the final output.

\begin{algorithm}[!t]
\caption{\rctreetracing}\label{alg:rctreetracing}
\DontPrintSemicolon
\SetKwFor{parfor}{parfor}{do}{}%
\AlgoDisplayBlockMarkers\SetAlgoBlockMarkers{}{end}%
Compute the RC-Tree $RCT$.\; \label{rctt:rct}
$\mb{bkts} \gets $ set of empty buckets corresponding to each $u \in V$
\parfor{each $e \in E$}{
    $u \gets$ $\rcnode$ associated to $e$\;
    $u \gets u.\mb{parent}$\;
    $f \gets$ edge associated to $u$\;
    \While{$r_f < r_e$ and $u$ is not the root}{
        $u \gets u.\mb{parent}$\;
        $f\gets$ edge associated to $u$\;
    }
    Add $e$ to bucket corresponding to $u$\;
}
\parfor{each $u \in V$}{
    Let $\mb{bkt}$ be the bucket associated to $\rcnode(u)$\;
    Sort the edges in $\mb{bkt}$ by ranks\;
    Let $e \gets $ edge associated to $\rcnode(u)$\;
    \parfor{each $i=1$ to $\mb{bkt.size}-1$}{
        $p[\mb{bkt}[i]] = \mb{bkt}[i+1]$\;
    }
    \If{$u$ is not root}{
        $p[\mb{bkt}[\mb{bkt.size}]] \gets e$\;
    }
}
\end{algorithm}

\cref{alg:rctreetracing} gives pseudocode for the \rctt{} algorithm.
It first computes the RC-tree associated to the tree contraction performed by \sldtreecon, without computing the output SLD or maintaining any spines (\cref{rctt:rct}).
We note that when a vertex (cluster) is compressed in the RC-tree, it will merge with the neighbor along the lesser rank edge, as required by the algorithm \sldtreecon.

Given the RC-tree, consider some edge $e$. Recall that each edge $e$ is associated to the $\rcnode$ of some vertex $v$ that gets raked or compressed via $e$. 
From the viewpoint of \sldtreecon, $\rcnode(v)$ corresponds to the stage when $e$ gets introduced into some heap. 
Edge $e$ will successively be involved in every rake/compress operation involving the cluster containing $v$ until (if at all) it gets protected (or filtered during a heap-filter operation) by some edge $f$. 
More specifically, $e$ is either protected by the first edge $f$ it encounters during tree contraction, after it's introduction, such that $r_e < r_f$, or it doesn't encounter such an edge and is, consequently, present in the heap at the root $\rcnode$, or in other words, it is protected at the root $\rcnode$. 
Let $u$ denote the $\rcnode$ where edge $e$ gets protected.
A critical observation here is that the set of rake/compress operations that includes $e$ until it becomes protected are associated to the $\rcnode$(s) along the (unique) path between $\rcnode(u)$ and $\rcnode(v)$. 
This is true due to the properties of RC-trees: the set of clusters containing the edge $e$ throughout tree contraction correspond exactly to the set of $\rcnode$s on the path from $\rcnode(v)$ until the root (in order). 
But, in \sldtreecon, $e$ gets filtered out from its heap once it encounters $f$ (or $\rcnode(u)$).

Thus, for each edge $e$, starting from $\rcnode(v)$, we traverse the $O(\log n)$ length path in the RC-tree along the path towards the root until we find $\rcnode(u)$. 
This way, for each $\rcnode$, we can collect all nodes $e$ that were protected at this node. 
This corresponds exactly to the set $S$ obtained by the first step via the \emph{heap-filter} operation; in case of the root, it is the remaining set of nodes in the spine. We can finally post-process these sets by sorting each of them by rank and updating the output SLD same as before. 
The pseudocode of this algorithm, which we refer to as \rctreetracing (\rctt in short), is given in \Cref{alg:rctreetracing}.

\myparagraph{Analysis and Implementation}
\rctreetracing is a simpler algorithm than \sldtreecon{} in the sense that it doesn't require meldable or filterable heaps and, in fact, doesn't require us to perform the actual merges of edges.
The main practical challenge in the algorithm is to maintain a dynamic adjacency list as the input tree contracts due to rakes and compresses.
The RC-tree can easily be computed in $O(n)$ work and $O(\log^2 n)$ depth.
Sorting within buckets to compute the final output runs in $O(n\log h)$ work and $O(\log^2n)$ depth, since the bucket sizes are $O(h)$ (all of these nodes are along some spine). 
The tree tracing step has the most work, i.e. $O(n\log n)$ work and $O(\log^2 n)$ depth, since it requires us to trace the entire height of the RC tree from each node in the worst case, and the height of the RC Tree is $O(\log n)$. However, in terms of experiments, we see that the tracing step is very fast; indeed the bottleneck is the RC tree construction time (see \cref{fig:breakdown}).

\section{Experimental Evaluation}\label{sec:experiments}

In this section we evaluate our parallel implementations for SLD construction and show the following main experimental results:

\begin{itemize}[topsep=0pt,itemsep=0pt,parsep=0pt,leftmargin=8pt]

\item \rctt{} is usually fastest on our inputs, achieving 2.1--132x speedup (16.9x geometric mean) over SeqUF on billion-scale inputs.

\item  \paruf{} achieves 2.1--150x speedup over SeqUF (5.92x geometric mean) on billion-scale inputs, but can be up to 151x slower than SeqUF on adversarial inputs.

\end{itemize}

\myparagraph{Experimental Setup}
Our experiments are performed on a 96-core
Dell PowerEdge R940 (with two-way hyperthreading) with $4\times 2.4\mbox{GHz}$
Intel 24-core 8160 Xeon processors (with 33MB L3 cache)
and 1.5\mbox{TB} of main memory.
Our programs use the work-stealing
scheduler provided by ParlayLib~\cite{parlay20}. 
Our programs are compiled with the
\texttt{g++} compiler (version 11.4) with the \texttt{-O3} flag.

\hide{
\myparagraph{Input/Output Structure}
The input to an SLD kernel is a weighted tree, $\mathcal{T}_{\mathsf{In}} = (V, E, w)$. 
The output to the kernel is a unique dendrogram $\mathcal{D}_{T}$; sequentially the dendrogram can be obtained by sorting the weights, and breaking any equal weights using the lexicographic ordering on the endpoints, i.e., sorting an edge $(u,v,w_{uv})$ by $(w_{uv}, \min(u,v), \max(u,v))$.
All SLD kernels studied in this paper deterministically produce the same output.
}


%
\myparagraph{Inputs}
The \emp{path} input is a path containing $n$ vertices and $n-1$ edges arranged in a path (or chain);
\emp{star} is a star on $n$ vertices where one vertex, the star center, has degree $n-1$, and all other vertices are connected to the center, and have degree $1$;
\emp{knuth} is a tree similar to the dependency structure of the Fischer-Yates-Knuth shuffle~\cite{blelloch2020randomized}, as follows: vertex $i > 0$ picks a neighbor in $[0, i-1]$ uniformly at random and connects itself to it.

We also generate several real-world tree inputs drawn real-world graphs. 
\defn{Friendster} is an
undirected graph describing friendships from a gaming network.\footnote{Source: \url{https://snap.stanford.edu/data/}.}
\defn{Twitter} is a directed graph of the Twitter network, where 
edges represent the follower
relationship~\cite{kwak2010twitter}.\footnote{Source: \url{http://law.di.unimi.it/webdata/twitter-2010/}.}
We build tree inputs for these real-world graphs by 
(1) symmetrizing them if needed
(2) setting the weight of each edge $(u,v)$ to be $\frac{1}{1 + t(u,v)}$, where $t(u,v)$ is the number of triangles incident on the edge $(u,v)$ and
(3) computing a minimum spanning tree.

We build another real-world tree input using the \defn{BigANN} dataset of SIFT image similarity descriptors; we compute the minimum spanning tree over an approximate $k$-nearest neighbor graph over 
a 100 million point subset of the BigANN dataset.\footnote{Source: \url{http://corpus-texmex.irisa.fr/}.}
For our construction, we used an in-memory version of DiskANN algorithm~\cite{diskann} implemented in the ParlayANN library~\cite{parlayann}.

\hide{
\myparagraph{Real-World Trees}
We also generate several real-world tree inputs drawn real-world graphs. We use the following real-world graphs:
\defn{com-DBLP (DB)} is a co-authorship network sourced from the
DBLP computer science bibliography. 
\footnote{Source: \url{https://snap.stanford.edu/data/com-DBLP.html}.}
\defn{YouTube (YT)} is a social-network formed by 
user-defined groups on the YouTube site.
\footnote{Source: \url{https://snap.stanford.edu/data/com-Youtube.html}.}
\defn{LiveJournal (LJ)} is a directed graph of the social network.
\footnote{Source: \url{https://snap.stanford.edu/data/soc-LiveJournal1.html}.}
\defn{com-Orkut (OK)} is an undirected
graph of the Orkut social network.
\defn{Friendster (FS)} is an
undirected graph describing friendships from a gaming network.
Both graphs are sourced from the SNAP dataset~\cite{leskovec2014snap}.\footnote{Source: \url{https://snap.stanford.edu/data/}.}
\defn{Twitter (TW)} is a directed graph of the Twitter network, where 
edges represent the follower
relationship~\cite{kwak2010twitter}.
\footnote{Source: \url{http://law.di.unimi.it/webdata/twitter-2010/}.}
\defn{ClueWeb (CW)} is a web graph from the Lemur project at CMU~\cite{boldi2004webgraph}.
\footnote{Source: \url{https://law.di.unimi.it/webdata/clueweb12/}.}
\defn{Hyperlink (HL)} is a hyperlink graph obtained from the
WebDataCommons dataset where nodes represent web pages~\cite{meusel15hyperlink}.
\footnote{Source: \url{http://webdatacommons.org/hyperlinkgraph/}.}
We use the undirected versions of all of the aforementioned graphs by symmetrizing all directed graphs.

Given a real-world graph, we construct a tree by weighting the edges in the graph using a weight scheme, and computing the minimum spanning tree of the graph using this scheme.
We use two weight schemes for this setup. \emp{InvDeg} computes the weight of a $(u,v)$ edge as $\frac{1}{d(u) + d(v)}$, where $d(u)$ is the degree of node $u$.
\emp{InvTri} sets the weight of a $(u,v)$ edge to be $\frac{1}{1 + T^{+}(u,v)}$ where $T^{+}(u,v)$ is the number of directed triangles closed by the edge $(u,v)$ in Latapy's compact-forward algorithm~\cite{latapy2008main}.\footnote{We use the number of directed triangles, since computing the edge-centric triangle-count for our largest web graphs (with over 200B edges) requires a prohibitive amount of memory and running time, even on a 1.5TB memory machine.}
}

\myparagraph{Weight Schemes}
We consider several different weight-schemes.
The unit \emp{unit} assigns all edges a weight of 1.
The \emp{perm} scheme generates a random permutation of the edges and assigns each edge a weight equal to its index in the random permutation.
The \emp{low-par} scheme is only applicable to paths, and is designed to be adversarial for the ParUF algorithm. This scheme assigns weights in increasing order to the first half of the edges in the path, and assigns weights in decreasing order for the second half of the path.

%

\begin{figure}[!t]
    \centering
    \vspace{-0.5em}
    \includegraphics[width=\columnwidth]{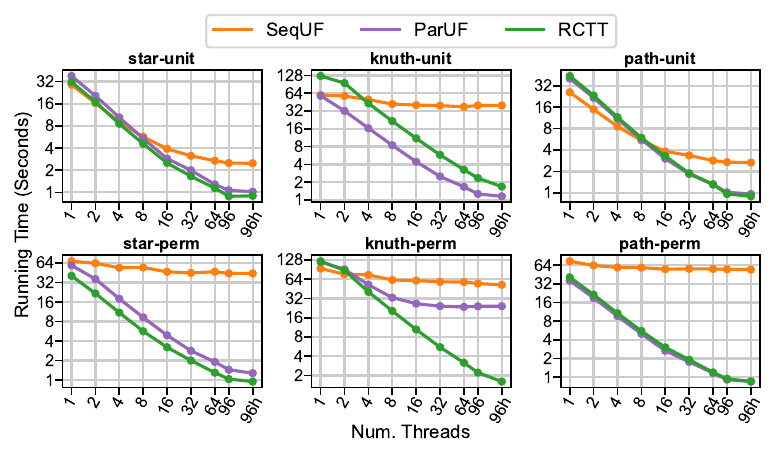}
    \caption{\small Running time of our SLD implementations on different input trees as a function of the number of threads. All inputs contain 100M vertices.
    \label{fig:speedup}}
    \Description{A plot displaying the running time of our SLD implementations on different input trees as a function of the number of threads. All inputs contain 100M vertices.}
    \vspace{-1.2em}
\end{figure}

\begin{figure}[!t]
    \centering
    \vspace{-1.2em}
    \includegraphics[width=\columnwidth]{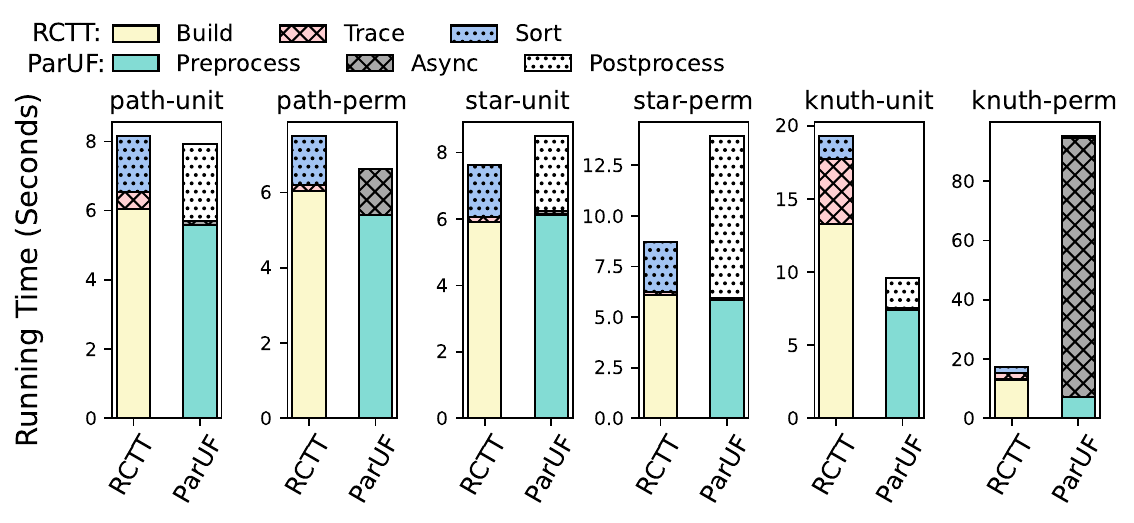}
    \caption{\small Parallel 192 thread running time breakdowns of the \rctt{} and \paruf{} algorithms on billion-scale inputs. 
    For \rctt{}, the Build step corresponds to building an RC tree; the Trace step finds the bucket associated with each edge; and the Sort step sorts each of the buckets.
    For \paruf{}, the Preprocess step sorts the edges and identifies initial local minima; the Async step performs bottom-up clustering; the Postprocess step sorts all remaining edges once the number of local-minima in the Async step becomes $1$.
    \label{fig:breakdown}}
    \Description{Parallel 192 thread running time breakdowns of the RCTT and ParUF algorithms on billion-scale inputs. 
    For RCTT, the Build step corresponds to building an RC tree; the Trace step finds the bucket associated with each edge; and the Sort step sorts each of the buckets.
    For ParUF, the Preprocess step sorts the edges and identifies initial local minima; the Async step performs bottom-up clustering; the Postprocess step sorts all remaining edges once the number of local-minima in the Async step becomes 1.}
    \vspace{-1.2em}
\end{figure}

\subsection{Algorithm Performance}
%

Next, we analyze the performance of our algorithms, including (1) their self-speedup, (2) their speedup over \sequf{}, and (3) the performance breakdown of our algorithms.
%
\cref{fig:speedup} shows the running times of our algorithms on a representative subset of the 100M-scale inputs as a function of the number of threads.

\myparagraph{\boldmath\sequf{}} SeqUF achieves between 1.36--11.6x self-relative speedup (2.94x geometric mean); it achieves the best self-speedups for the star and path graphs using unit weights as shown in \cref{fig:speedup}. 
We emphasize that despite its name, \sequf{} is able to leverage parallelism since the first step in the algorithm is to sort the edges, and we use a highly optimized parallel sort from ParlayLib~\cite{parlay20}.
For all other inputs, in which the input tree or edge weights induce an irregular access pattern, the algorithm achieves poor speedup (under 2x).
In more detail, for star and path graphs using unit weights, the edges merged all lie one after the other in memory, and so the access pattern of the sequential algorithm has good locality.
When the weights are permuted or the input tree has an irregular structure (e.g., in the case of the knuth inputs), the algorithm accesses essentially two random cache-lines in every iteration.

\begin{table}
\footnotesize
\centering
\vspace{-2.5em}
\caption{\small Parallel running times of our SLD implementations on different tree inputs. The last two columns show the speedup of our implementations over SeqUF. \label{table:synthetic}
}
\begin{tcolorbox}[colback=white, enhanced, tabularx={cr|ccc|cc}, width=0.9\linewidth]
Type  & \multicolumn{1}{c|}{Sizes} & SeqUF & ParUF & RCTT & $\frac{\text{SeqUF}}{\text{ParUF}}$ & $\frac{\text{SeqUF}}{\text{RCTT}}$\\
\midrule
\multirow{3}[2]{*}{path} 
      & 10M   & .253  & .120             & \underline{.101} & 2.10 & 2.50 \\
      & 100M  & 2.20  & .962             & \underline{.897} & 2.28 & 2.45 \\
      & 1B    & 21.8  & \underline{7.93} & 8.15             & 2.74  & 2.67 \\
\midrule
\multirow{3}[2]{*}{path perm} 
      & 10M   & 5.56  & \underline{.090} & .099 & 61.7  & 56.1 \\
      & 100M  & 68.6  & \underline{.881} & .904 & 77.8  & 75.8 \\
      & 1B    & 989   & \underline{6.61} & 7.49 & 149.6 & 132  \\
\midrule
\multirow{3}[2]{*}{path low-par} 
      & 10M   & .335  & 41.9  & \underline{.101}  & 0.007 & 3.31 \\
      & 100M  & 2.42  & 366   & \underline{.884}  & 0.006 & 2.73 \\
      & 1B    & 23.4  & 2640  & \underline{7.73}  & 0.008 & 3.02 \\
\midrule
\multirow{3}[2]{*}{star} 
      & 10M   & .252  & .125  & \underline{.111} & 2.01 & 2.27 \\
      & 100M  & 2.04  & 1.03  & \underline{.895} & 1.98 & 2.28 \\
      & 1B    & 20.3  & 8.49  & \underline{7.61} & 2.39 & 2.66 \\
\midrule
\multirow{3}[2]{*}{star perm} 
      & 10M   & 4.71  & .141  & \underline{.116} & 33.4 & 40.6 \\
      & 100M  & 56.1  & 1.29  & \underline{1.01} & 43.4 & 55.5 \\
      & 1B    & 824   & 13.9  & \underline{8.68} & 59.2 & 94.9 \\
\midrule
\multirow{3}[2]{*}{knuth}
      & 10M   & 1.70  & \underline{.140} & .156 & 12.1 & 10.9 \\
      & 100M  & 39.2  & \underline{1.12} & 1.69 & 35.0 & 23.1 \\
      & 1B    & 458   & \underline{9.61} & 19.2 & 47.6 & 23.8 \\
\midrule
\multirow{3}[2]{*}{knuth perm}
      & 10M   & 5.83  & 2.55  & \underline{.155} & 2.28 & 37.6 \\
      & 100M  & 79.9  & 37.7  & \underline{1.61} & 2.11 & 49.6 \\
      & 1B    & 1110  & 95.1  & \underline{17.3} & 11.6 & 64.1 \\
\end{tcolorbox}
\end{table}

\myparagraph{\boldmath\paruf{}} 
\paruf{} achieves between 4.91--50.1x self-relative speedup (30.1x geometric mean). 
In \cref{fig:speedup}, it achieves the lowest speedups on the knuth input with permuted weights due to this input tree resulting in a high height dendrogram ($h = \text{2.5M}$) which is not amenable to our post-processing optimization.
\cref{fig:breakdown}, which shows the performance breakdown of \paruf{} on different inputs shows that on the knuth input with permuted weights, almost all of the time is spent on the asynchronous Union-Find step (the while loop starting on \cref{paruf:while} in \cref{alg:activation_based}).
Although several other inputs have high height (e.g., knuth with unit weights, whose dendrogram forms a path of length $n-1$), they are amenable to the post-processing optimization, and thus \paruf{} achieves good speedup since the post-processing step simply sorts the remaining edges.
\paruf{} typically begins to out-perform \sequf{} with more than 8 threads.
Compared to \sequf{}, as shown in \cref{table:synthetic}, it obtains between 2.1---150x speedup over \sequf{} (5.92x geometric mean speedup) on the billion-scale inputs; however, it performs poorly on the adversarial low-parallelism input (path low-par) it is 151x worse at the 100M scale.

\begin{figure}[!t]
    \centering
    \vspace{-0.5em}
    \includegraphics[width=\columnwidth]{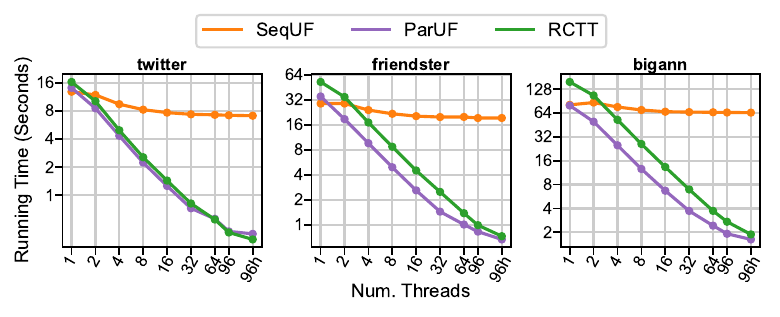}
    \caption{\small Running time of our SLD implementations on different real-world input trees as a function of the number of threads.
    \label{fig:speedup_rw}}
    \Description{A plot displaying the running time of our SLD implementations on different real-world input trees as a function of the number of threads.}
    \vspace{-1.2em}
\end{figure}

\myparagraph{\boldmath\rctt{}} 
The \rctt{} algorithm achieves the most consistent speedups 
among the algorithms studied in this paper, achieving between 35.2--75.5x self-speedup (52.1x geometric mean).
From \cref{fig:speedup} and \cref{table:synthetic}, we can see that \rctt{} is usually the fastest algorithm when using all threads, and is always within a factor of 2x of the performance of \paruf{}.
Similar to \paruf{}, it starts to outperform \sequf{} on all inputs after about 8 threads, and is never slower than \sequf{} on any input, at any of the scales that we evaluated.
Unlike \paruf{}, whose behavior depends on the amount of parallelism available to it and sometimes performs much worse than \sequf{}, \rctt{} always obtains speedups over \sequf{}, achieving between 2.1---132x speedup (16.9x geometric mean speedup) over \sequf{} on the billion-scale inputs.
We can see from \cref{fig:breakdown} that despite the {\em Trace} step being the most costly step in theory (see \cref{sec:rctt}), it takes at most 23\% of the time across our inputs, and is usually only a few percentage of the total running time.
The majority of the time is spent on the RC tree construction; optimizing this step by designing faster tree contraction algorithms is an interesting direction for future work.

\myparagraph{Real-World Inputs}
We also ran our implementations on three real-world tree inputs described earlier in this section (results in \Cref{fig:speedup_rw}).
On these inputs, we observe that \sequf{} achieves modest speedups more similar to the permuted weights, rather than the high self-relative speedup in the unit weight case.
In particular, it achieves between 1.2--1.8x self-relative speedup.
On the other hand, both \paruf{} and \sequf{} achieve strong self-speedups, with \paruf{} achieving between 36--52x self-speedup and \rctt{} achieving between 48.7--84x self-speedup.
Both of our parallel algorithms achieve strong speedups over \sequf{}---on all 192 threads \paruf{} is between 18.4--39.8x faster, and \rctt{} is between 21.1--34.4x faster.

\section{Conclusion}
In this paper, we gave optimal parallel algorithms
for computing the single-linkage dendrogram.
We described a framework for obtaining merge-based 
divide-and-conquer algorithms, and instantiated the
framework using parallel tree contraction, showing that it
yields an optimal work deterministic algorithm with 
poly-logarithmic depth.
We also designed two practical algorithms, \paruf{} and
\rctt{}, both of which have provable guarantees on their work 
and depth, and which achieve strong speedups over a 
highly optimized sequential baseline.
An interesting question is whether we can extend our 
approach to obtain good dynamic algorithms for 
maintaining the single-linkage dendrogram.


\begin{acks}
This work is supported by NSF grants CCF-2103483 and CNS-2317194, NSF CAREER Award CCF-2339310, the UCR Regents Faculty Award, and the Google Research Scholar Program. We thank the anonymous reviewers for their useful comments.
\end{acks}

\bibliographystyle{ACM-Reference-Format}
\bibliography{bibliography/main}


\begin{thebibliography}{44}


\ifx \showCODEN    \undefined \def \showCODEN     #1{\unskip}     \fi
\ifx \showDOI      \undefined \def \showDOI       #1{#1}\fi
\ifx \showISBNx    \undefined \def \showISBNx     #1{\unskip}     \fi
\ifx \showISBNxiii \undefined \def \showISBNxiii  #1{\unskip}     \fi
\ifx \showISSN     \undefined \def \showISSN      #1{\unskip}     \fi
\ifx \showLCCN     \undefined \def \showLCCN      #1{\unskip}     \fi
\ifx \shownote     \undefined \def \shownote      #1{#1}          \fi
\ifx \showarticletitle \undefined \def \showarticletitle #1{#1}   \fi
\ifx \showURL      \undefined \def \showURL       {\relax}        \fi
\providecommand\bibfield[2]{#2}
\providecommand\bibinfo[2]{#2}
\providecommand\natexlab[1]{#1}
\providecommand\showeprint[2][]{arXiv:#2}

\bibitem[Abrahamson et~al\mbox{.}(1989)]%
        {abrahamson1989simple}
\bibfield{author}{\bibinfo{person}{Karl Abrahamson}, \bibinfo{person}{Norm Dadoun}, \bibinfo{person}{David~G. Kirkpatrick}, {and} \bibinfo{person}{T Przytycka}.} \bibinfo{year}{1989}\natexlab{}.
\newblock \showarticletitle{A simple parallel tree contraction algorithm}.
\newblock \bibinfo{journal}{\emph{Journal of Algorithms}} \bibinfo{volume}{10}, \bibinfo{number}{2} (\bibinfo{year}{1989}), \bibinfo{pages}{287--302}.
\newblock


\bibitem[Acar et~al\mbox{.}(2017)]%
        {acar2017brief}
\bibfield{author}{\bibinfo{person}{Umut~A Acar}, \bibinfo{person}{Vitaly Aksenov}, {and} \bibinfo{person}{Sam Westrick}.} \bibinfo{year}{2017}\natexlab{}.
\newblock \showarticletitle{Brief Announcement: Parallel Dynamic Tree Contraction via Self-Adjusting Computation}. In \bibinfo{booktitle}{\emph{{ACM} Symposium on Parallelism in Algorithms and Architectures (SPAA)}}.
\newblock


\bibitem[Anderson(2023)]%
        {anderson2023parallel}
\bibfield{author}{\bibinfo{person}{Daniel Anderson}.} \bibinfo{year}{2023}\natexlab{}.
\newblock \emph{\bibinfo{title}{Parallel Batch-Dynamic Algorithms Dynamic Trees, Graphs, and Self-Adjusting Computation}}.
\newblock \bibinfo{thesistype}{Ph.\,D. Dissertation}. \bibinfo{school}{Carnegie Mellon University}.
\newblock


\bibitem[Baron(2019)]%
        {baron2019machine}
\bibfield{author}{\bibinfo{person}{Dalya Baron}.} \bibinfo{year}{2019}\natexlab{}.
\newblock \bibinfo{title}{Machine Learning in Astronomy: a practical overview}.
\newblock
\newblock
\showeprint[arxiv]{1904.07248}~[astro-ph.IM]


\bibitem[Benz{\'e}cri(1982)]%
        {benzecri1982construction}
\bibfield{author}{\bibinfo{person}{J-P Benz{\'e}cri}.} \bibinfo{year}{1982}\natexlab{}.
\newblock \showarticletitle{Construction d'une classification ascendante hi{\'e}rarchique par la recherche en cha{\^\i}ne des voisins r{\'e}ciproques}.
\newblock \bibinfo{journal}{\emph{Cahiers de l'analyse des donn{\'e}es}} \bibinfo{volume}{7}, \bibinfo{number}{2} (\bibinfo{year}{1982}), \bibinfo{pages}{209--218}.
\newblock


\bibitem[Blelloch et~al\mbox{.}(2020a)]%
        {parlay20}
\bibfield{author}{\bibinfo{person}{Guy~E. Blelloch}, \bibinfo{person}{Daniel Anderson}, {and} \bibinfo{person}{Laxman Dhulipala}.} \bibinfo{year}{2020}\natexlab{a}.
\newblock \showarticletitle{{ParlayLib} - A Toolkit for Parallel Algorithms on Shared-Memory Multicore Machines}. In \bibinfo{booktitle}{\emph{{ACM} Symposium on Parallelism in Algorithms and Architectures (SPAA)}}. \bibinfo{pages}{507–509}.
\newblock
\urldef\tempurl%
\url{https://cmuparlay.github.io/parlaylib/}
\showURL{%
\tempurl}


\bibitem[Blelloch et~al\mbox{.}(2021)]%
        {blelloch18notes}
\bibfield{author}{\bibinfo{person}{Guy~E. Blelloch}, \bibinfo{person}{Laxman Dhulipala}, {and} \bibinfo{person}{Yihan Sun}.} \bibinfo{year}{2021}\natexlab{}.
\newblock \bibinfo{title}{Introduction to Parallel Algorithms}.
\newblock \bibinfo{howpublished}{\url{https://www.cs.cmu.edu/~guyb/paralg/paralg/parallel.pdf}}.
\newblock
\newblock
\shownote{Carnegie Mellon University}.


\bibitem[Blelloch et~al\mbox{.}(2020b)]%
        {blelloch2019optimal}
\bibfield{author}{\bibinfo{person}{Guy~E. Blelloch}, \bibinfo{person}{Jeremy~T. Fineman}, \bibinfo{person}{Yan Gu}, {and} \bibinfo{person}{Yihan Sun}.} \bibinfo{year}{2020}\natexlab{b}.
\newblock \showarticletitle{Optimal Parallel Algorithms in the Binary-Forking Model}. In \bibinfo{booktitle}{\emph{{ACM} Symposium on Parallelism in Algorithms and Architectures (SPAA)}}. \bibinfo{pages}{89–102}.
\newblock
\showISBNx{9781450369350}


\bibitem[Blelloch et~al\mbox{.}(2020c)]%
        {blelloch2020randomized}
\bibfield{author}{\bibinfo{person}{Guy~E. Blelloch}, \bibinfo{person}{Yan Gu}, \bibinfo{person}{Julian Shun}, {and} \bibinfo{person}{Yihan Sun}.} \bibinfo{year}{2020}\natexlab{c}.
\newblock \showarticletitle{Randomized Incremental Convex Hull is Highly Parallel}. In \bibinfo{booktitle}{\emph{{ACM} Symposium on Parallelism in Algorithms and Architectures (SPAA)}}. \bibinfo{pages}{103–115}.
\newblock


\bibitem[Campello et~al\mbox{.}(2015)]%
        {campello2015hierarchical}
\bibfield{author}{\bibinfo{person}{Ricardo~JGB Campello}, \bibinfo{person}{Davoud Moulavi}, \bibinfo{person}{Arthur Zimek}, {and} \bibinfo{person}{J{\"o}rg Sander}.} \bibinfo{year}{2015}\natexlab{}.
\newblock \showarticletitle{Hierarchical density estimates for data clustering, visualization, and outlier detection}.
\newblock \bibinfo{journal}{\emph{ACM Transactions on Knowledge Discovery from Data (TKDD)}} \bibinfo{volume}{10}, \bibinfo{number}{1} (\bibinfo{year}{2015}), \bibinfo{pages}{1--51}.
\newblock


\bibitem[Cormen et~al\mbox{.}(2009)]%
        {CLRS}
\bibfield{author}{\bibinfo{person}{Thomas~H. Cormen}, \bibinfo{person}{Charles~E. Leiserson}, \bibinfo{person}{Ronald~L. Rivest}, {and} \bibinfo{person}{Clifford Stein}.} \bibinfo{year}{2009}\natexlab{}.
\newblock \bibinfo{booktitle}{\emph{Introduction to Algorithms}}.
\newblock


\bibitem[Demaine et~al\mbox{.}(2009)]%
        {demaine2009cartesian}
\bibfield{author}{\bibinfo{person}{Erik~D Demaine}, \bibinfo{person}{Gad~M Landau}, {and} \bibinfo{person}{Oren Weimann}.} \bibinfo{year}{2009}\natexlab{}.
\newblock \showarticletitle{On Cartesian trees and range minimum queries}. In \bibinfo{booktitle}{\emph{Automata, Languages and Programming: 36th International Colloquium, ICALP 2009, Rhodes, Greece, July 5-12, 2009, Proceedings, Part I 36}}. Springer, \bibinfo{pages}{341--353}.
\newblock


\bibitem[Dhulipala et~al\mbox{.}(2021)]%
        {dhulipala2021hierarchical}
\bibfield{author}{\bibinfo{person}{Laxman Dhulipala}, \bibinfo{person}{David Eisenstat}, \bibinfo{person}{Jakub {\L}{\k{a}}cki}, \bibinfo{person}{Vahab Mirrokni}, {and} \bibinfo{person}{Jessica Shi}.} \bibinfo{year}{2021}\natexlab{}.
\newblock \showarticletitle{Hierarchical agglomerative graph clustering in nearly-linear time}. In \bibinfo{booktitle}{\emph{International Conference on Machine Learning}}. PMLR, \bibinfo{pages}{2676--2686}.
\newblock


\bibitem[Feigelson and Babu(1998)]%
        {feigelson1998statistical}
\bibfield{author}{\bibinfo{person}{E.D. Feigelson} {and} \bibinfo{person}{G.J. Babu}.} \bibinfo{year}{1998}\natexlab{}.
\newblock \showarticletitle{Statistical Methodology for Large Astronomical Surveys}.
\newblock \bibinfo{journal}{\emph{Symposium - International Astronomical Union}}  \bibinfo{volume}{179} (\bibinfo{year}{1998}), \bibinfo{pages}{363–370}.
\newblock


\bibitem[Gasperini et~al\mbox{.}(2019)]%
        {gasperini2019genome}
\bibfield{author}{\bibinfo{person}{Molly Gasperini}, \bibinfo{person}{Andrew~J Hill}, \bibinfo{person}{Jos{\'e}~L McFaline-Figueroa}, \bibinfo{person}{Beth Martin}, \bibinfo{person}{Seungsoo Kim}, \bibinfo{person}{Melissa~D Zhang}, \bibinfo{person}{Dana Jackson}, \bibinfo{person}{Anh Leith}, \bibinfo{person}{Jacob Schreiber}, \bibinfo{person}{William~S Noble}, {et~al\mbox{.}}} \bibinfo{year}{2019}\natexlab{}.
\newblock \showarticletitle{A genome-wide framework for mapping gene regulation via cellular genetic screens}.
\newblock \bibinfo{journal}{\emph{Cell}} \bibinfo{volume}{176}, \bibinfo{number}{1} (\bibinfo{year}{2019}), \bibinfo{pages}{377--390}.
\newblock


\bibitem[Gazit et~al\mbox{.}(1988)]%
        {gazit1988optimal}
\bibfield{author}{\bibinfo{person}{Hillel Gazit}, \bibinfo{person}{Gary~L Miller}, {and} \bibinfo{person}{Shang-Hua Teng}.} \bibinfo{year}{1988}\natexlab{}.
\newblock \showarticletitle{Optimal tree contraction in the EREW model}.
\newblock In \bibinfo{booktitle}{\emph{Concurrent Computations: Algorithms, Architecture, and Technology}}. \bibinfo{pages}{139--156}.
\newblock


\bibitem[G{\"o}tz et~al\mbox{.}(2018)]%
        {gotz2018parallel}
\bibfield{author}{\bibinfo{person}{Markus G{\"o}tz}, \bibinfo{person}{Gabriele Cavallaro}, \bibinfo{person}{Thierry G{\'e}raud}, \bibinfo{person}{Matthias Book}, {and} \bibinfo{person}{Morris Riedel}.} \bibinfo{year}{2018}\natexlab{}.
\newblock \showarticletitle{Parallel computation of component trees on distributed memory machines}.
\newblock \bibinfo{journal}{\emph{{IEEE} Transactions on Parallel and Distributed Systems}} \bibinfo{volume}{29}, \bibinfo{number}{11} (\bibinfo{year}{2018}), \bibinfo{pages}{2582--2598}.
\newblock


\bibitem[Gower and Ross(1969)]%
        {Gower1969MST}
\bibfield{author}{\bibinfo{person}{J.~C. Gower} {and} \bibinfo{person}{G.~J.~S. Ross}.} \bibinfo{year}{1969}\natexlab{}.
\newblock \showarticletitle{Minimum Spanning Trees and Single Linkage Cluster Analysis}.
\newblock \bibinfo{journal}{\emph{Journal of the Royal Statistical Society. seriesx C (Applied Statistics)}} \bibinfo{volume}{18}, \bibinfo{number}{1} (\bibinfo{year}{1969}), \bibinfo{pages}{54--64}.
\newblock
\showISSN{00359254, 14679876}


\bibitem[Gu et~al\mbox{.}(2015)]%
        {GSSB15}
\bibfield{author}{\bibinfo{person}{Yan Gu}, \bibinfo{person}{Julian Shun}, \bibinfo{person}{Yihan Sun}, {and} \bibinfo{person}{Guy~E. Blelloch}.} \bibinfo{year}{2015}\natexlab{}.
\newblock \showarticletitle{A Top-Down Parallel Semisort}. In \bibinfo{booktitle}{\emph{{ACM} Symposium on Parallelism in Algorithms and Architectures (SPAA)}}. \bibinfo{pages}{24–34}.
\newblock
\showISBNx{9781450335881}


\bibitem[Hajiaghayi et~al\mbox{.}(2022)]%
        {hajiaghayi2022adaptive}
\bibfield{author}{\bibinfo{person}{MohammadTaghi Hajiaghayi}, \bibinfo{person}{Marina Knittel}, \bibinfo{person}{Hamed Saleh}, {and} \bibinfo{person}{Hsin-Hao Su}.} \bibinfo{year}{2022}\natexlab{}.
\newblock \showarticletitle{{Adaptive Massively Parallel Constant-Round Tree Contraction}}. In \bibinfo{booktitle}{\emph{13th Innovations in Theoretical Computer Science Conference (ITCS 2022)}}, \bibfield{editor}{\bibinfo{person}{Mark Braverman}} (Ed.), Vol.~\bibinfo{volume}{215}. \bibinfo{pages}{83:1--83:23}.
\newblock
\showISBNx{978-3-95977-217-4}
\showISSN{1868-8969}


\bibitem[Havel et~al\mbox{.}(2019)]%
        {havel2019efficient}
\bibfield{author}{\bibinfo{person}{Ji{\v{r}}{\'\i} Havel}, \bibinfo{person}{Fran{\c{c}}ois Merciol}, {and} \bibinfo{person}{S{\'e}bastien Lef{\`e}vre}.} \bibinfo{year}{2019}\natexlab{}.
\newblock \showarticletitle{Efficient tree construction for multiscale image representation and processing}.
\newblock \bibinfo{journal}{\emph{Journal of Real-Time Image Processing}}  \bibinfo{volume}{16} (\bibinfo{year}{2019}), \bibinfo{pages}{1129--1146}.
\newblock


\bibitem[Hendrix et~al\mbox{.}(2013)]%
        {hendrix2013scalable}
\bibfield{author}{\bibinfo{person}{William Hendrix}, \bibinfo{person}{Diana Palsetia}, \bibinfo{person}{Md~Mostofa~Ali Patwary}, \bibinfo{person}{Ankit Agrawal}, \bibinfo{person}{Wei-keng Liao}, {and} \bibinfo{person}{Alok Choudhary}.} \bibinfo{year}{2013}\natexlab{}.
\newblock \showarticletitle{A scalable algorithm for single-linkage hierarchical clustering on distributed-memory architectures}. In \bibinfo{booktitle}{\emph{IEEE Symposium on Large-Scale Data Analysis and Visualization (LDAV)}}. IEEE, \bibinfo{pages}{7--13}.
\newblock


\bibitem[Henry et~al\mbox{.}(2005)]%
        {henry2005cluster}
\bibfield{author}{\bibinfo{person}{David~B Henry}, \bibinfo{person}{Patrick~H Tolan}, {and} \bibinfo{person}{Deborah Gorman-Smith}.} \bibinfo{year}{2005}\natexlab{}.
\newblock \showarticletitle{Cluster analysis in family psychology research.}
\newblock \bibinfo{journal}{\emph{Journal of Family Psychology}} \bibinfo{volume}{19}, \bibinfo{number}{1} (\bibinfo{year}{2005}), \bibinfo{pages}{121}.
\newblock


\bibitem[JaJa(1992)]%
        {JaJa92}
\bibfield{author}{\bibinfo{person}{J. JaJa}.} \bibinfo{year}{1992}\natexlab{}.
\newblock \bibinfo{booktitle}{\emph{Introduction to Parallel Algorithms}}.
\newblock


\bibitem[Kwak et~al\mbox{.}(2010)]%
        {kwak2010twitter}
\bibfield{author}{\bibinfo{person}{Haewoon Kwak}, \bibinfo{person}{Changhyun Lee}, \bibinfo{person}{Hosung Park}, {and} \bibinfo{person}{Sue Moon}.} \bibinfo{year}{2010}\natexlab{}.
\newblock \showarticletitle{What is Twitter, a social network or a news media?}. In \bibinfo{booktitle}{\emph{International World Wide Web Conference (WWW)}}. \bibinfo{pages}{591–600}.
\newblock
\showISBNx{9781605587998}


\bibitem[Letunic and Bork(2007)]%
        {letunic2007interactive}
\bibfield{author}{\bibinfo{person}{Ivica Letunic} {and} \bibinfo{person}{Peer Bork}.} \bibinfo{year}{2007}\natexlab{}.
\newblock \showarticletitle{Interactive Tree Of Life (iTOL): an online tool for phylogenetic tree display and annotation}.
\newblock \bibinfo{journal}{\emph{Bioinformatics}} \bibinfo{volume}{23}, \bibinfo{number}{1} (\bibinfo{year}{2007}), \bibinfo{pages}{127--128}.
\newblock


\bibitem[Manning et~al\mbox{.}(2008)]%
        {irbook}
\bibfield{author}{\bibinfo{person}{Christopher~D Manning}, \bibinfo{person}{Prabhakar Raghavan}, {and} \bibinfo{person}{Hinrich Sch{\"u}tze}.} \bibinfo{year}{2008}\natexlab{}.
\newblock \bibinfo{booktitle}{\emph{Introduction to Information Retrieval}}.
\newblock


\bibitem[Manohar et~al\mbox{.}(2024)]%
        {parlayann}
\bibfield{author}{\bibinfo{person}{Magdalen~Dobson Manohar}, \bibinfo{person}{Zheqi Shen}, \bibinfo{person}{Guy Blelloch}, \bibinfo{person}{Laxman Dhulipala}, \bibinfo{person}{Yan Gu}, \bibinfo{person}{Harsha~Vardhan Simhadri}, {and} \bibinfo{person}{Yihan Sun}.} \bibinfo{year}{2024}\natexlab{}.
\newblock \showarticletitle{ParlayANN: Scalable and Deterministic Parallel Graph-Based Approximate Nearest Neighbor Search Algorithms}. In \bibinfo{booktitle}{\emph{{ACM} Symposium on Principles and Practice of Parallel Programming (PPOPP)}}. \bibinfo{pages}{270–285}.
\newblock
\showISBNx{9798400704352}


\bibitem[Miller and Reif(1985)]%
        {miller1985parallel}
\bibfield{author}{\bibinfo{person}{Gary~L Miller} {and} \bibinfo{person}{John~H Reif}.} \bibinfo{year}{1985}\natexlab{}.
\newblock \showarticletitle{Parallel tree contraction and its application}. In \bibinfo{booktitle}{\emph{{IEEE} Symposium on Foundations of Computer Science (FOCS)}}, Vol.~\bibinfo{volume}{26}. \bibinfo{pages}{478--489}.
\newblock


\bibitem[Moschini et~al\mbox{.}(2017)]%
        {moschini2017hybrid}
\bibfield{author}{\bibinfo{person}{Ugo Moschini}, \bibinfo{person}{Arnold Meijster}, {and} \bibinfo{person}{Michael~HF Wilkinson}.} \bibinfo{year}{2017}\natexlab{}.
\newblock \showarticletitle{A hybrid shared-memory parallel max-tree algorithm for extreme dynamic-range images}.
\newblock \bibinfo{journal}{\emph{IEEE transactions on pattern analysis and machine intelligence}} \bibinfo{volume}{40}, \bibinfo{number}{3} (\bibinfo{year}{2017}), \bibinfo{pages}{513--526}.
\newblock


\bibitem[Nolet et~al\mbox{.}(2023)]%
        {nolet2023cuslink}
\bibfield{author}{\bibinfo{person}{Corey~J. Nolet}, \bibinfo{person}{Divye Gala}, \bibinfo{person}{Alex Fender}, \bibinfo{person}{Mahesh doixjade}, \bibinfo{person}{Joe Eaton}, \bibinfo{person}{Edward Raff}, \bibinfo{person}{John Zedlewski}, \bibinfo{person}{Brad Rees}, {and} \bibinfo{person}{Tim Oates}.} \bibinfo{year}{2023}\natexlab{}.
\newblock \showarticletitle{cuSLINK: Single-Linkage Agglomerative Clustering on the GPU}. In \bibinfo{booktitle}{\emph{ECML PKDD}}. \bibinfo{pages}{711–726}.
\newblock
\showISBNx{978-3-031-43411-2}


\bibitem[Ouzounis(2020)]%
        {ouzounis2020segmentation}
\bibfield{author}{\bibinfo{person}{Georgios~K Ouzounis}.} \bibinfo{year}{2020}\natexlab{}.
\newblock \showarticletitle{Segmentation strategies for the alpha-tree data structure}.
\newblock \bibinfo{journal}{\emph{Pattern Recognition Letters}}  \bibinfo{volume}{129} (\bibinfo{year}{2020}), \bibinfo{pages}{232--239}.
\newblock


\bibitem[Ouzounis and Soille(2012)]%
        {ouzounis2012alpha}
\bibfield{author}{\bibinfo{person}{Georgios~K Ouzounis} {and} \bibinfo{person}{Pierre Soille}.} \bibinfo{year}{2012}\natexlab{}.
\newblock \showarticletitle{The alpha-tree algorithm}.
\newblock \bibinfo{journal}{\emph{JRC Scientific and Policy Report}} (\bibinfo{year}{2012}).
\newblock


\bibitem[Reif and Tate(1994)]%
        {reif94treecontraction}
\bibfield{author}{\bibinfo{person}{John~H. Reif} {and} \bibinfo{person}{Stephen~R. Tate}.} \bibinfo{year}{1994}\natexlab{}.
\newblock \showarticletitle{Dynamic parallel tree contraction (extended abstract)}. In \bibinfo{booktitle}{\emph{{ACM} Symposium on Parallelism in Algorithms and Architectures (SPAA)}}. \bibinfo{pages}{114–121}.
\newblock
\showISBNx{0897916719}


\bibitem[Sao et~al\mbox{.}(2024)]%
        {sao2024pandora}
\bibfield{author}{\bibinfo{person}{Piyush Sao}, \bibinfo{person}{Andrey Prokopenko}, {and} \bibinfo{person}{Damien Lebrun-Grandié}.} \bibinfo{year}{2024}\natexlab{}.
\newblock \bibinfo{title}{PANDORA: A Parallel Dendrogram Construction Algorithm for Single Linkage Clustering on GPU}.
\newblock
\newblock
\showeprint[arxiv]{2401.06089}~[cs.LG]


\bibitem[Sch{\"u}tze et~al\mbox{.}(2008)]%
        {schutze2008introduction}
\bibfield{author}{\bibinfo{person}{Hinrich Sch{\"u}tze}, \bibinfo{person}{Christopher~D Manning}, {and} \bibinfo{person}{Prabhakar Raghavan}.} \bibinfo{year}{2008}\natexlab{}.
\newblock \bibinfo{booktitle}{\emph{Introduction to information retrieval}}.
\newblock


\bibitem[Shun and Blelloch(2014)]%
        {SB14cartesian}
\bibfield{author}{\bibinfo{person}{Julian Shun} {and} \bibinfo{person}{Guy~E. Blelloch}.} \bibinfo{year}{2014}\natexlab{}.
\newblock \showarticletitle{A simple parallel cartesian tree algorithm and its application to parallel suffix tree construction}.
\newblock \bibinfo{journal}{\emph{ACM Trans. Parallel Comput.}} \bibinfo{volume}{1}, \bibinfo{number}{1} (\bibinfo{year}{2014}).
\newblock
\showISSN{2329-4949}


\bibitem[Shun et~al\mbox{.}(2015)]%
        {shun2014sequential}
\bibfield{author}{\bibinfo{person}{Julian Shun}, \bibinfo{person}{Yan Gu}, \bibinfo{person}{Guy~E. Blelloch}, \bibinfo{person}{Jeremy~T. Fineman}, {and} \bibinfo{person}{Phillip~B. Gibbons}.} \bibinfo{year}{2015}\natexlab{}.
\newblock \showarticletitle{Sequential random permutation, list contraction and tree contraction are highly parallel}. In \bibinfo{booktitle}{\emph{{ACM-SIAM} Symposium on Discrete Algorithms (SODA)}}. \bibinfo{pages}{431–448}.
\newblock


\bibitem[Subramanya et~al\mbox{.}(2019)]%
        {diskann}
\bibfield{author}{\bibinfo{person}{Suhas~Jayaram Subramanya}, \bibinfo{person}{Devvrit}, \bibinfo{person}{Rohan Kadekodi}, \bibinfo{person}{Ravishankar Krishaswamy}, {and} \bibinfo{person}{Harsha~Vardhan Simhadri}.} \bibinfo{year}{2019}\natexlab{}.
\newblock \showarticletitle{DiskANN: fast accurate billion-point nearest neighbor search on a single node}. In \bibinfo{booktitle}{\emph{Neural Information Processing Systems (NeurIPS)}}.
\newblock


\bibitem[Sumengen et~al\mbox{.}(2021)]%
        {sumengen2021scaling}
\bibfield{author}{\bibinfo{person}{Baris Sumengen}, \bibinfo{person}{Anand Rajagopalan}, \bibinfo{person}{Gui Citovsky}, \bibinfo{person}{David Simcha}, \bibinfo{person}{Olivier Bachem}, \bibinfo{person}{Pradipta Mitra}, \bibinfo{person}{Sam Blasiak}, \bibinfo{person}{Mason Liang}, {and} \bibinfo{person}{Sanjiv Kumar}.} \bibinfo{year}{2021}\natexlab{}.
\newblock \bibinfo{title}{Scaling Hierarchical Agglomerative Clustering to Billion-sized Datasets}.
\newblock
\newblock
\showeprint[arxiv]{2105.11653}~[cs.LG]


\bibitem[Wang et~al\mbox{.}(2021)]%
        {WangEtAl21}
\bibfield{author}{\bibinfo{person}{Yiqiu Wang}, \bibinfo{person}{Shangdi Yu}, \bibinfo{person}{Yan Gu}, {and} \bibinfo{person}{Julian Shun}.} \bibinfo{year}{2021}\natexlab{}.
\newblock \showarticletitle{Fast Parallel Algorithms for Euclidean Minimum Spanning Tree and Hierarchical Spatial Clustering}. In \bibinfo{booktitle}{\emph{Proceedings of the 2021 International Conference on Management of Data}}. \bibinfo{pages}{1982–1995}.
\newblock
\showISBNx{9781450383431}


\bibitem[Yengo et~al\mbox{.}(2022)]%
        {yengo2022saturated}
\bibfield{author}{\bibinfo{person}{Lo{\"\i}c Yengo}, \bibinfo{person}{Sailaja Vedantam}, \bibinfo{person}{Eirini Marouli}, \bibinfo{person}{Julia Sidorenko}, \bibinfo{person}{Eric Bartell}, \bibinfo{person}{Saori Sakaue}, \bibinfo{person}{Marielisa Graff}, \bibinfo{person}{Anders~U Eliasen}, \bibinfo{person}{Yunxuan Jiang}, \bibinfo{person}{Sridharan Raghavan}, {et~al\mbox{.}}} \bibinfo{year}{2022}\natexlab{}.
\newblock \showarticletitle{A saturated map of common genetic variants associated with human height}.
\newblock \bibinfo{journal}{\emph{Nature}} \bibinfo{volume}{610}, \bibinfo{number}{7933} (\bibinfo{year}{2022}), \bibinfo{pages}{704--712}.
\newblock


\bibitem[Yim and Ramdeen(2015)]%
        {yim2015hierarchical}
\bibfield{author}{\bibinfo{person}{Odilia Yim} {and} \bibinfo{person}{Kylee~T Ramdeen}.} \bibinfo{year}{2015}\natexlab{}.
\newblock \showarticletitle{Hierarchical cluster analysis: comparison of three linkage measures and application to psychological data}.
\newblock \bibinfo{journal}{\emph{The Quantitative Methods for Psychology}} \bibinfo{volume}{11}, \bibinfo{number}{1} (\bibinfo{year}{2015}), \bibinfo{pages}{8--21}.
\newblock


\bibitem[Yu et~al\mbox{.}(2021)]%
        {yu2021parchain}
\bibfield{author}{\bibinfo{person}{Shangdi Yu}, \bibinfo{person}{Yiqiu Wang}, \bibinfo{person}{Yan Gu}, \bibinfo{person}{Laxman Dhulipala}, {and} \bibinfo{person}{Julian Shun}.} \bibinfo{year}{2021}\natexlab{}.
\newblock \showarticletitle{ParChain: a framework for parallel hierarchical agglomerative clustering using nearest-neighbor chain}.
\newblock \bibinfo{journal}{\emph{Proc. VLDB Endow.}} \bibinfo{volume}{15}, \bibinfo{number}{2} (\bibinfo{year}{2021}), \bibinfo{pages}{285–298}.
\newblock
\showISSN{2150-8097}


\end{thebibliography}

\appendix

\section{Related Work}\label{sec:related_work}
Single-linkage clustering has been studied for over half
a century, starting with the early work of Gower and Ross~\cite{Gower1969MST}.
Since then, it has found widespread application in a variety
of scientific disciplines and industrial applications~\cite{yengo2022saturated, gasperini2019genome, letunic2007interactive,ouzounis2012alpha, gotz2018parallel, havel2019efficient, baron2019machine, feigelson1998statistical, henry2005cluster, yim2015hierarchical, irbook}.

\myparagraph{Single-Linkage Dendrogram Algorithms}
In the past decade, due to the importance of single-linkage 
clustering, serious algorithmic consideration of the core
problem of computing a single-linkage dendrogram began.
The work of Demaine et al.~\cite{demaine2009cartesian} gives
an algorithm showing that if the cost of sorting the
edges is ``free'', SLD can be solved in $O(n)$ time. Their algorithm
is based on an nice argument using decremental tree connectivity.
Their linear-work bound is not directly comparable to our results
since they assume the edges are sorted (and thus bypass
comparison lower bounds).

In more recent years, different communities have studied the
SLD problem, and other related hierarchical tree building problems.
A large body of work has come out of the image analysis community,
where SLD is studied under the moniker of ``alpha-tree'' algorithms~\cite{ouzounis2012alpha, ouzounis2020segmentation, moschini2017hybrid, gotz2018parallel}. Unfortunately the algorithms are highly specific to analyzing 2D and 3D data, and parallel algorithms in this domain~\cite{moschini2017hybrid, gotz2018parallel} are not work-efficient and typically not rigorously analyzed in parallel models.

The most relevant related work is the paper of Wang et al.~\cite{WangEtAl21} who recently gave the first work-efficient
parallel algorithm for SLD.
Their algorithm
is randomized and computes the SLD in $O(n \log n)$ expected work 
and $O(\log^2 n \log\log n)$ depth with high probability (\whp{}).
Although their algorithm is work-efficient with
respect to \sequf{}, it 
relies on applying divide-and-conquer over the weights, which is 
implemented using the Euler Tour Technique~\cite{JaJa92}.
Based on private communication with the authors, we understand that due to its complicated nature, this algorithm does not consistently outperform the simple \sequf{} algorithm. 
The authors only released the code for \sequf{} and suggested to always use \sequf{} rather than the theoretically-efficient algorithm.
The algorithm is randomized due to the use of semisort~\cite{WangEtAl21, GSSB15}, and there is no clear way to derandomize it to obtain a deterministic parallel algorithm.

Our paper gives two algorithms that have better work than the algorithm of Wang et al.~\cite{WangEtAl21}, but have worse depth bounds since the depth of our tree contraction algorithm is $O(\log^2 n \log^2 h) = \Omega(\log^2 n (\log\log n)^2)$. 
However, our third algorithm  (\rctt{}) provides a strict improvement over their algorithm, while also being simple and deterministic, since \rctt{} achieves $O(n \log n)$ work and $O(\log^2 n)$ depth in the binary-forking model. 
We note that in the PRAM model, the \rctt{} algorithm requires $O(n\log n)$ work and $O(\log n)$ depth. 
Improving \rctt{} to obtain $O(n\log h)$ work and $O(\log n)$ depth or $O(n)$ work and $\mathsf{polylog}(n)$ depth if the edges are pre-sorted are two interesting directions for future work.

\section{Lower Bound}\label{sec:lb}

\lowerbound*
\begin{proof}
For a given value of $h$, we build a tree with $n/h$ connected
components, each with $h$ elements. The goal of each component
is to solve an independent instance of sorting, which has a
comparison-sorting lower bound of $\Omega(h \log h)$ work~\cite{CLRS} to solve. To create an input tree for our comparison-based SLD algorithm to solve, we connect all elements within each component into a star.
It is not hard to see that after solving SLD, the elements in each star will be totally ordered based on their rank, and so solving SLD on the aforementioned input tree will solve all of the sorting instances.
Since each sorting instance requires $\Omega(h \log h)$ comparisons (work), the $n/h$ instances requires $\Omega(n \log h)$ work in total to solve.
\end{proof}

\end{document}
\endinput